\documentclass[journal,twoside,web]{ieeecolor}
\usepackage{generic}
\usepackage{cite}

\usepackage{amsmath,amssymb,amsfonts,amsthm}
\usepackage{algorithmicx}
\usepackage{graphicx}
\usepackage{algorithm}
\usepackage{hyperref}
\hypersetup{hidelinks}
\usepackage{textcomp}
\usepackage{subfig}
\usepackage{colortbl,booktabs}
\newtheoremstyle{boldstyle} % Style name
{\topsep}   % Space above
{\topsep}   % Space below
{\normalfont}  % Body font (not italic)
{}          % Indent amount
{\bfseries} % Theorem head font (bold)
{.}         % Punctuation after theorem head
{ }         % Space after theorem head
{}          % Theorem head spec (empty = default)

\theoremstyle{boldstyle}

\newtheorem{definition}{Definition}
\newtheorem{assumption}{Assumption}
\newtheorem{corollary}{Corollary}
\newtheorem{theorem}{Theorem}
\newtheorem{lemma}{Lemma}
\newtheorem{proposition}{Proposition}
\newtheorem{remark}{Remark}

\def\BibTeX{{\rm B\kern-.05em{\sc i\kern-.025em b}\kern-.08em
    T\kern-.1667em\lower.7ex\hbox{E}\kern-.125emX}}
\begin{document}
\title{Predictive control barrier functions for piecewise affine systems with non-smooth constraints}
\author{Kanghui He, Anil Alan, Shengling Shi, \IEEEmembership{Member, IEEE}, Ton van den Boom, and Bart De Schutter, \IEEEmembership{Fellow, IEEE}
	\thanks{This project has received funding from the European Research Council (ERC) under the European Union's Horizon 2020 research and innovation programme (Grant agreement No.\ 101018826 - ERC Advanced Grant CLariNet), and by the Rubicon Postdoctoral Fellowship (Correspondence No. 2026/ENW/02250137), funded by the Netherlands Organisation for Scientific Research (NWO).}
	\thanks{Kanghui He is affiliated with the Department of Engineering Science, University of Oxford, OX1 3PJ Oxford, U.K.,        
    email: kanghui.he@eng.ox.ac.uk. Anil Alan, Shengling Shi, Ton van den Boom, and Bart De Schutter are with Delft Center for Systems and Control, Delft University of Technology, Delft, The Netherlands, e-mail: $\{$a.alan, shengling.shi, a.j.j.vandenBoom, b.deschutter$\}$@tudelft.nl.}}

\maketitle
\begin{abstract}
Obtaining control barrier functions (CBFs) with large safe sets for complex nonlinear systems and constraints is a challenging task. Predictive CBFs address this issue by using an online finite-horizon optimal control problem that implicitly defines a large safe set. The optimal control problem, also known as the predictive safety filter (PSF), involves predicting the system’s flow under a given backup control policy. However, for non-smooth systems and constraints, some key elements, such as CBF gradients and the sensitivity of the flow, are not well-defined, making the current methods inadequate for ensuring safety. Additionally, for control-non-affine systems, the PSF is generally nonlinear and non-convex, posing challenges for real-time computation. This paper considers piecewise affine systems, which are usually control-non-affine, under nonlinear state and polyhedral input constraints. We solve the safety issue by incorporating set-valued generalized Clarke derivatives in the PSF design. We show that enforcing CBF constraints across all elements of the generalized Clarke derivatives suffices to guarantee safety. Moreover, to lighten the computational overhead, we propose an explicit approximation of the PSF. The resulting control methods are demonstrated through numerical examples.
\end{abstract}


\section{Introduction}

\subsection{Background}
Motivated by emerging safety-critical applications such as autonomous driving \cite{alan2023control} and automatic therapeutic regimens \cite{jarrett2020optimal}, control of nonlinear systems subject to state and input constraints has received growing attention in the academic community. Various safe control methods have been developed, such as model predictive control (MPC) \cite{borrelli2017predictive}, control barrier functions (CBFs)-based safety filters \cite{ames2016control}, solving Hamilton-Jacobi (HJ) reachability problems \cite{bansal2017hamilton}, and constrained reinforcement learning \cite{he2024approximate}. All methods have distinct advantages but also inherent limitations that hinder their practical effectiveness. In particular, MPC can be computationally expensive for hybrid and nonlinear systems; HJ reachability suffers from the curse of dimensionality; and constrained reinforcement learning often lacks rigorous safety guarantees.  

CBFs provide a modular framework with formal safety guarantees, called the CBF-based safety filter, to monitor and modify potentially unsafe control actions online. CBFs use their super-level sets to define control invariant sets, and the safety filter further involves a constrained optimization problem with CBF-based safety constraints to produce a safe control policy that keeps the system within the control invariant set. However, computing a permissive CBF, i.e., one associated with a non-conservative control invariant set can be particularly challenging for nonlinear systems with both state and input constraints. As a result, CBF-based safety filters may be overly conservative.

Drawing inspiration from MPC, researchers have proposed to use model predictions to reduce the conservatism of CBFs both in discrete-time \cite{wabersich2022predictive,wabersich2021predictive} and in continuous-time domains \cite{gurriet2020scalable}, \cite{chen2021backup,van2024disturbance,breeden2022predictive}. In particular, they use a finite-horizon constrained optimal control problem with a terminal CBF constraint to implicitly represent a larger control invariant set, resulting in a predictive safety filter (PSF). The terminal CBF is referred to as the backup CBF in \cite{gurriet2020scalable}. For continuous-time control-affine systems, the PSF can be cast as convex quadratic programs (QPs), enabling efficient real-time implementation.

It should be emphasized that existing predictive CBF approaches for continuous-time systems are inherently limited to smooth systems and smooth safety constraints \cite{chen2021backup,van2024disturbance,breeden2022predictive,gurriet2020scalable}. This requirement arises from the need to differentiate both the safety constraints and the system’s flow in the predictive CBF formulation. Furthermore, to allow for fast computation, these approaches concentrate on control-affine systems. For control-non-affine systems, the resulting safety filter usually involves solving a time-consuming non-convex nonlinear optimization problem.

Piecewise affine (PWA) systems can approximate general nonlinear systems to arbitrary precision and naturally capture hybrid behaviors \cite{bemporad1999control}. On the other hand, PWA systems are generally control-non-affine and non-smooth. Therefore, the above-mentioned CBF-based safety filters cannot be directly adopted to efficiently provide safe control inputs for PWA systems. 

\subsection{Contributions}
\begin{figure*}
	\centering
	\includegraphics[width=400pt,clip]{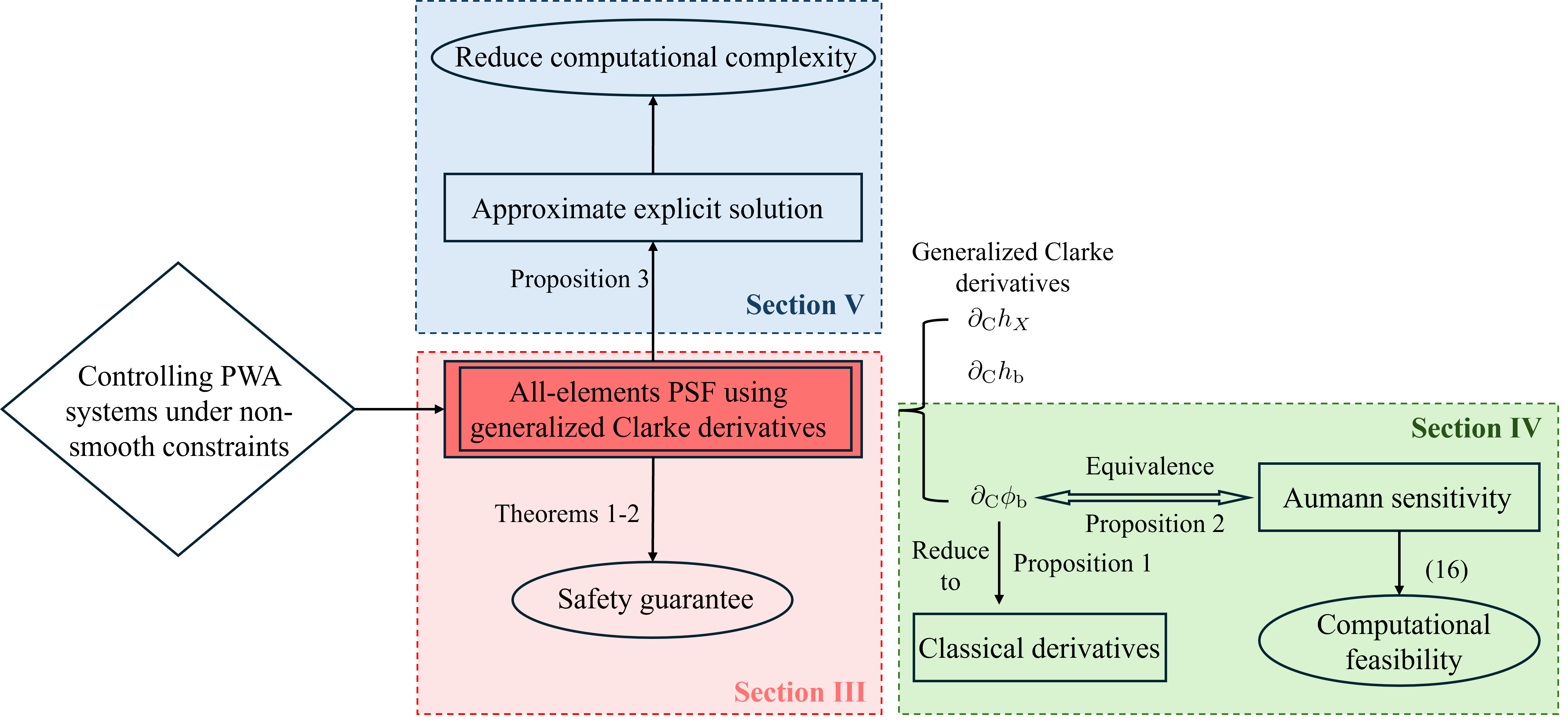}
	\caption{Graphical representation of the paper structure.}
	\label{structure}
\end{figure*}

In this paper, we address the challenge of designing PSFs for continuous-time piecewise affine (PWA) systems subject to arbitrary, possibly non-smooth, Lipschitz continuous state constraints\footnote{We say that a constraint $x\in X$ is continuous if $X$ can be represented by the level set of a continuous function, i.e., $X =\{x| f(x) \leq 0\}$, where $f$ is continuous.} and polyhedral input constraints. The resulting PSF is formulated as a mixed-integer quadratic programming (MIQP) problem. The main contributions of this paper, in comparison to the state of the art, are:

1. We propose a novel PSF that uses generalized Clarke derivatives to enforce safety constraints. This framework is applicable to any Lipschitz continuous but non-smooth CBF and any continuous PWA system with Lipschitz continuous but non-smooth state constraints. We rigorously analyze the safety guarantees of the resulting control policy (Lemma \ref{lemma_invariance} and Theorem \ref{theoreom_forward}). 

2. We present a tractable approach for computing the generalized Clarke derivative (Clarke sensitivity) of the system's flow. In particular, we introduce an alternative Aumann sensitivity and show their equivalence. Notably, the Aumann sensitivity can be efficiently computed through numerical integration of the system.

3. To reduce the computational complexity of solving the MIQP corresponding to the PSF online, we introduce an explicit approximation of the PSF. This approximation is provably feasible and significantly lowers the online computational burden.

4. Additionally, the theoretical safety result (Lemma 2) generalizes classical CBF conditions that require continuous differentiability to encompass CBFs that are merely Lipschitz continuous. This extension applies to general nonlinear systems, not just PWA ones.

The paper’s main results and their organization are shown in Fig. \ref{structure}.

\subsection{Related work}

\textbf{(Approximate) MPC for PWA systems.} Hybrid MPC is widely used to control PWA systems under linear state and input constraints \cite{borrelli2017predictive}. In the hybrid MPC formulation, the PWA system is equivalently transformed into a mixed-logical dynamical (MLD) system that includes both continuous and discrete state and input variables \cite{bemporad1999control}. This transformation allows the MPC problem to be formulated as a mixed-integer convex quadratic or linear optimization problem, which is solved online to generate safe control inputs. Explicit MPC is an offline version of hybrid MPC. In explicit MPC, the state-to-solution map is computed offline by multi-parametric programming \cite{bemporad2000piecewise}. However, MPC methods struggle with significant computational challenges in high-dimensional systems, either due to solving complex mixed-integer problems or evaluating intricate explicit control laws.

This limitation of MPC motivates the research of approximating MPC control laws, by using, e.g., PWA simplicial functions \cite{poggi2011explicit}, neural networks \cite{he2024approximate}, sampling-based polytopic trees \cite{sadraddini2019sampling}, and minimum-time objective functions \cite{grieder2005stabilizing}. Another promising approach to simplify the MPC problem is to employ reinforcement learning \cite{da2024integrating} or supervised learning \cite{mallick2024learning,russo2023learning} to determine the values of the discrete decision variables, effectively reducing the mixed-integer problem to a convex optimization problem. However, while exact MPC offers formal performance guarantees, it leads to a high computational burden, especially for nonlinear, non-convex constraints. Approximate MPC, on the other hand, still has open challenges regarding stability, robustness, and sample complexity.

\textbf{Predictive safety filters.} Research in both the continuous-time and discrete-time domains has contributed to the development of PSFs, which aim to reduce the conservatism of standard safety filters. Representative work includes \cite{wabersich2022predictive,wabersich2021predictive} in the discrete-time domain and \cite{chen2021backup,van2024disturbance,breeden2022predictive,gurriet2020scalable} (sometimes also called backup CBFs) in the continuous-time domain. Based on \cite{wabersich2022predictive}, the authors in \cite{didier2023approximate,he2023state,he2025learning} approximate discrete-time predictive CBFs using neural networks to reduce the computational burden of safety filters. The work in \cite{molnar2023safety} extends the PSFs of \cite{gurriet2020scalable,van2024disturbance} to guarantee safe behavior for complex systems by designing safety filters for reduced-order models. However, continuous-time predictive CBF methods are fundamentally limited to smooth systems and smooth safety constraints. As mentioned before, this requirement arises from the need to differentiate both the safety constraints and the system’s flow in the predictive CBF formulation, which requires computing the Jacobian of the closed-loop dynamics with a smooth backup controller. Moreover, to date, there is no related work addressing the design of PSFs for PWA systems.

\textbf{Non-smooth CBFs.} \cite{glotfelter2017nonsmooth} proposes non-smooth CBFs for deterministic multi-agent systems, leveraging generalized gradients to ensure forward invariance. This result has been employed in applications where non-smooth safety constraints naturally arise, such as obstacle avoidance involving ellipsoids and polytopes \cite{thirugnanam2023nonsmooth} and network connectivity maintenance \cite{ong2023nonsmooth}. Besides, formulating multiple constraints with “min” and “max” functions will generally result in a non-smooth CBF. Building on \cite{glotfelter2017nonsmooth}, the work in \cite{thirugnanam2023nonsmooth} extends non-smooth CBFs to stochastic systems. Related to these techniques, piecewise continuous CBFs are introduced in \cite{cortes2024discontinuous}, where the authors of \cite{cortes2024discontinuous} argue that, in addition to satisfying the CBF inequality, a tangent cone condition at the boundary of the safe set should also be enforced to ensure forward invariance. Meanwhile, Cohen et al. \cite{cohen2023characterizing} take a different perspective by smoothing the solution to CBF-based safety filters using the Implicit Function Theorem, rather than focusing on the smoothness of the CBFs themselves. All of the above approaches focus on classical CBFs with no predictions. This means that the analysis primarily addresses the non-smoothness of the CBF itself. In contrast, the non-smoothness in our setting arises not only from the structure of the CBF but also from state constraints and the PWA dynamics, making it more challenging to derive safety guarantees.

\subsection{Notations} The set of $n$-dimensional real vectors is denoted by $\mathbb{R}^n$. The sets $\mathbb{N}$ and $\mathbb{N}^+$ represent the set of non-negative integers and the set of positive integers, respectively. Furthermore, $\mathbb{N}_a= \{0,1,...,a\}$ and $\mathbb{N}^+_a= \{1,2,...,a\}$. The norm $||x||$ is the Euclidean norm of the vector $x$, and $||A||$ is the induced norm of the Euclidean norm for the matrix $A$. The notation $\mathrm{int}(S)$ refers to the interior of the set $S$, and $\mathrm{clo}(S)$ is the closure of the set $S$. Besides, $1_n$ is the vector of ones with $n$ elements, and $I_n$ is the identity matrix with dimension $n\times n$. A continuous function $\alpha:  [0, \infty) \to [0, \infty]$ is said to belong to class $\mathcal{K}_\infty$ if it is strictly increasing and if $\alpha(0)=0$ and $\lim _{r \to \infty} \alpha(r)=\infty $. The step function $\mathrm{step}(\cdot)$ returns 1 if the input is non-negative, and 0 otherwise. The notation $D f(x_0) \in \mathbb{R}^{m\times n}$ represents the derivative of the differentiable function $f(\cdot): \mathbb{R}^n \to \mathbb{R}^m$ at $x_0$. For any matrix $C$ and vector $c$, let $(C)_i$ and $(c)_i$ denote the $i$-th row of $C$ and the $i$-th entry of $c$, respectively.
%can work in conjunction with an arbitrary learning algorithm.
\section{Preliminaries and backup CBFs}

Consider a general continuous-time nonlinear system of the form
\begin{equation}\label{pwa}
	\dot x = f(x,u),
\end{equation}
where $x \in \mathbb{R}^n$ is the state, $u \in \mathbb{R}^m$ is the input, and $f(\cdot):\mathbb{R}^{n+m} \to \mathbb{R}^n$ is locally Lipschitz continuous. The system must always satisfy the given state constraint $x \in X:=\{x\in \mathbb{R}^n| h_X(x) \geq 0\}$ and input constraint $u \in U$. Here, the function $h_X(\cdot): \mathbb{R}^n \to \mathbb{R}$ is locally Lipschitz continuous, and $U$ is a polyhedron\footnote{While the results of this paper can be extended to cases where $U$ is a union of polyhedra, we assume for simplicity that $U$ is a single polyhedron}. 

\subsection{Control barrier functions}
CBFs offer an effective approach for designing safe control systems.
\begin{definition}[Control barrier function\cite{ames2019control}]
	A continuously differentiable function $h(\cdot): \mathbb{R}^n \to \mathbb{R}$ is a \emph{control barrier function} for \eqref{pwa} if there exists an $\alpha \in \mathcal{K}_\infty$ such that 
	\begin{equation}\label{cbf}
		\sup _{u \in U} \dot{h}(x, u) =\sup _{u \in U} \{\nabla h(x)\cdot f(x,u)\} \geq -\alpha(h(x))
	\end{equation}
	holds for any $x \in S:= \{x \in \mathbb{R}^n | h(x) \geq 0\}$.
\end{definition}

With a CBF available, there always exists a controller $\pi(\cdot): \mathbb{R}^n \to \mathbb{R}^m$ satisfying $\pi(x) \in U$ and $\nabla h(x) \cdot  f(x,\pi(x)) \geq -\alpha(h(x))$ for all $x \in S$. This controller will render $S$ forward invariant for \eqref{pwa}, and $S$ is called the safe set \cite{ames2019control}. Under the supplemental condition $S \subseteq X$, the controller $\pi$ makes the system always satisfy the constraints if the initial state is within $S$. 

\subsection{Predictive safety filters for smooth systems}

For nonlinear systems, it is nontrivial to get a permissive CBF, i.e., a CBF with a large safe set $S$ \cite{wabersich2023data}. Drawing ideas from model predictive control, PSFs address this issue using model predictions \cite{molnar2023safety,chen2021backup}. We recall this method in this subsection.

First, a CBF $h_\mathrm{b}: \mathbb{R}^n \to \mathbb{R}$ and an associated safe controller $\pi_\mathrm{b}$, satisfying (i) $\pi_\mathrm{b}(x) \in U$ for all $x\in X$, and (ii) $\nabla h_\mathrm{b}(x) \cdot f(x,\pi_\mathrm{b}(x)) \geq -\alpha_\mathrm{b}(h_\mathrm{b}(x))$ for all $x \in S_\mathrm{b}:=\{x \in \mathbb{R}^n | h_\mathrm{b}(x) \geq 0\} \subseteq X$, are obtained based on existing methods such as LMI or sum-of-squares programming \cite{wang2018permissive}, depending on the form of the system. Here, the subscript “b” refers to “backup”. 

The closed-loop system with the backup controller can be written as
\begin{equation}\label{closed_loop}
	\dot x = f_\mathrm{b} (x): = f(x,\pi_\mathrm{b}(x))
\end{equation}
If $f$ and $\pi_\mathrm{b}$ are locally Lipschitz continuous, which is commonly the case in many practical scenarios, including linear, PWA, and polynomial systems and policies, then for any $x(0) = x_0 \in  \mathbb{R}^n$, according to \cite[Theorem 3.2]{khalil2002nonlinear}, the system \eqref{closed_loop} has a unique solution $x(t) = \phi_\mathrm{b}(x_0,t)$, $\forall t\in [0,\infty)$.

\begin{definition}[Constrained reachable set]
	Given the system \eqref{closed_loop} and the set $X$, for any set $S$ and any time $T>0$, the $T$-step \emph{constrained reachable set} for \eqref{closed_loop} is defined as $\Phi_\mathrm{const}(S,T) :=\{x \in \mathbb{R}^n |  \phi_\mathrm{b}(x,\tau) \in X,\;\forall \tau \in [0,T] \wedge \phi_\mathrm{b}(x,T) \in S \}$.
\end{definition}

 The PSF method uses an important property related to constrained reachable sets, stating that the constrained reachable sets of any invariant set $S$ are always controlled-invariant and are supersets of $S$. See \cite[Lemma 1]{molnar2023safety} and \cite[Proposition 2]{gurriet2020scalable} for a more precise statement. This invariance property allows us to synthesize a safe controller that has a larger safe region. 

\begin{lemma}[Safety of backup CBFs \cite{molnar2023safety,gurriet2020scalable}]\label{lemma_safety}
	Consider the system \eqref{pwa} with the state and input constraints $x \in X$ and $u \in U$, the backup controller $\pi_\mathrm{b}$, and the backup CBF $h_\mathrm{b}$. Suppose that $f$, $h_X$, $h_\mathrm{b}$, and $\pi_\mathrm{b}$ are continuously differentiable and that $X$ is compact. Then, there exist a controller $\pi_\mathrm{safe} (\cdot): \mathbb{R}^n \to \mathbb{R}^m$ and $\alpha, \alpha_{\mathrm{b}} \in \mathcal{K}_{\infty}$ such that the following conditions hold  $\forall T \geq 0$ and $\forall x \in \Phi_\mathrm{const}(S_\mathrm{b},T)$: 
	\begin{align}\label{cbf condition}
		\dot{h}_X\left(\phi_{\mathrm{b}}( x,\tau), \pi_\mathrm{safe}(x)\right) & \geq-\alpha\left(h_X\left(\phi_{\mathrm{b}}( x,\tau)\right)\right), \forall \tau \in[0, T] \nonumber \\
		\dot{h}_{\mathrm{b}}\left(\phi_{\mathrm{b}}(x,T), \pi_\mathrm{safe}(x)\right) & \geq-\alpha_{\mathrm{b}}\left(h_{\mathrm{b}}\left(\phi_{\mathrm{b}}(x,T)\right)\right) \nonumber\\
		\pi_\mathrm{safe}(\phi_{\mathrm{b}}(x,\tau)) &\in U, \forall \tau \in[0, T].
	\end{align}
	Moreover, any locally Lipschitz continuous controller $\pi_\mathrm{safe}$ that satisfies \eqref{cbf condition} renders $\Phi_\mathrm{const}(S_\mathrm{b},T)$ forward invariant for \eqref{pwa}.
\end{lemma}

The proof of Lemma \ref{lemma_safety} relies on verifying that 
\begin{equation}\label{CBF_candidate}
		h(x)\!:=\!\min \left\{\min _{\tau \in[0,T]} h_X(\phi_\mathrm{b}(x, \tau)), h_\mathrm{b}(\phi_\mathrm{b}(x, T))\right\},
	\end{equation}
    is a valid CBF. We call the function $h$ in \eqref{CBF_candidate} the \emph{predictive CBF}. 

Lemma \ref{lemma_safety} can be directly employed for safe controller synthesis by designing the following optimization-based PSF\footnote{The conditions under which $\pi_\mathrm{safe}$ is locally Lipschitz, which is required in Lemma \ref{lemma_safety}, can be found in \cite{MESTRES2025101098}.}:
\begin{align}\label{safety filter}
	&\pi_\mathrm{safe} (x)\nonumber\\
	=&\underset{u \in U}{\operatorname{argmin}}\left\|u-\pi_{\mathrm{r}}(x)\right\|^2\nonumber\\	
	&\mathrm{s.t.}\;\dot{h}_X\left(\phi_{\mathrm{b}}(x,\tau), u\right)  \geq-\alpha\left(h_X\left(\phi_{\mathrm{b}}( x,\tau)\right)\right), \forall \tau \in[0, T] \nonumber \\
	&\quad \;\;\dot{h}_{\mathrm{b}}\left(\phi_{\mathrm{b}}( x,T), u\right)  \geq-\alpha_{\mathrm{b}}\left(h_{\mathrm{b}}\left(\phi_{\mathrm{b}}( x,T)\right)\right),
\end{align}
where $T \geq 0$ is the prediction horizon and $\pi_\mathrm{r}(\cdot) :\mathbb{R}^n \to \mathbb{R}^m$ is a reference controller. Problem \eqref{safety filter} is, in general, nonlinear and non-convex. In existing literature, a control-affine system and a polyhedral $U$ are usually considered, which makes \eqref{safety filter} a convex quadratic program (QP) (after discretizing $[0,\;T]$) \cite{van2024disturbance,chen2021backup,molnar2023safety}. Besides, thanks to the existence of $\pi_\mathrm{b}$, problem \eqref{safety filter} is guaranteed to be feasible $\forall x \in \Phi_\mathrm{const}(S_\mathrm{b},T)$. 

For the continuously differentiable system \eqref{pwa}, the derivatives $\dot{h}_X$ and $\dot{h}_{\mathrm{b}}$ in \eqref{safety filter} can be expressed as
\begin{align}\label{sensitivity}
	\dot{h}_X\left(\phi_{\mathrm{b}}( x,\tau), u\right) & =D h_X\left(\phi_{\mathrm{b}}( x,\tau)\right) \frac{\partial \phi_{\mathrm{b}}( x,\tau)}{\partial x} f(x,u) \\
	\dot{h}_{\mathrm{b}}\left(\phi_{\mathrm{b}}(x,T), u\right) & =D h_{\mathrm{b}}\left(\phi_{\mathrm{b}}( x,T)\right) \frac{\partial \phi_{\mathrm{b}}(x,T)}{\partial x}f(x,u) ,
\end{align}
where $\frac{\partial \phi_\mathrm{b}( x,\tau)}{\partial  x} \in \mathbb{R}^{n\times n}$ is called the sensitivity of the flow to its initial condition\footnote{In \eqref{safety filter}, $x$ refers to the initial condition of the safety filter, which is also the current state of the system \eqref{pwa}.}. 

\begin{remark}
	Note that the backup controller $\pi_\mathrm{b}$ is a virtual controller that is never applied. It is only used to derive the CBF constraints in \eqref{safety filter}. The applied controller is $\pi_\mathrm{safe}$. 
\end{remark}

\subsection{Problem statement}

In this paper, we focus on a class of continuous PWA systems:
\begin{equation}\label{pwa_form}
	\dot x = f_\mathrm{PWA}(x,u):=A_i x + B_i u + c_i, \text{ for } \left[\begin{array}{l}
		x\\u
	\end{array} \right]\in \mathcal{P}_i,
\end{equation}
where $f_\mathrm{PWA}:\mathbb{R}^{n+m} \to \mathbb{R}^n$ is a continuous PWA function, $\{\mathcal{P}_i\}^p_{i=1}$ constitutes a strict polyhedral partition \cite[Definition 4.5]{borrelli2017predictive} of the state-input space. The union of all boundaries of all $\mathcal{P}_i,\;i\in\mathbb{N}^+_p$ is called the boundary of the PWA system. 

Computing the time derivatives in \eqref{sensitivity} is an essential step for designing the safety filter \eqref{safety filter}. It should be noted that existing literature usually assumes that the system, the CBF, and the constraint $X$ are continuously differentiable \cite{van2024disturbance,chen2021backup,molnar2023safety}. In this case, the terms on the right-hand side of \eqref{sensitivity} are well-defined and safety performance is guaranteed (Lemma \ref{lemma_safety}). However, the PWA system is usually not continuously differentiable at the boundary of each $\mathcal{P}_i$, and many commonly encountered constraints, such as polytopic constraints, are not differentiable everywhere. This means that the time derivatives in \eqref{sensitivity} may have multiple possible values and that the results in Lemma \ref{lemma_safety} may not hold. 

Secondly, for control-affine systems, problem \eqref{safety filter} can be formulated as a convex QP, which is computationally inexpensive to solve in real time. However, as observed in \eqref{pwa_form}, $f_\mathrm{PWA}(x,u)$ is not globally affine in $u$. In this case, problem \eqref{safety filter} becomes an MIQP, which needs much more computational effort to solve than a convex QP.

In the remainder, the following problems will be addressed.

\textbf{Problem 1} (addressed in Section \ref{section_synthesis}): How to design a PSF with safety guarantees for a given PWA system, which should address the non-smoothness of the PWA dynamics, the state constraint, and the CBF $h_\mathrm{b}$?

\textbf{Problem 2} (addressed in Section \ref{section_nonsmooth}): How to overcome the computational intractability of the critical component of the safety filter, in particular, the not well-defined sensitivity for PWA systems? 

\textbf{Problem 3} (addressed in Section \ref{section_approximate}): How to approximate the solution to the safety filter in a computationally efficient way to enable swift online implementation?

\subsection{Generalized sensitivity}
For non-smooth functions, the concept of the Jacobian matrix from smooth functions can be generalized using several approaches, such as the generalized Clarke derivative \cite{clarke1976inverse} and Mordukhovich subdifferential \cite{Mordukhovich03062017}. In this paper, we advocate for the use of the generalized Clarke derivative due to its convenience in performance analysis. The generalized Clarke derivative extends the idea of the Jacobian matrix to Lipschitz continuous functions that may not be differentiable everywhere.

\begin{definition}[Generalized Clarke derivative \cite{clarke1976inverse}]\label{Clarke definition}
	The \emph{generalized Clarke derivative} of a locally Lipschitz continuous function $\phi(\cdot):\mathbb{R}^n \to \mathbb{R}^{n_\phi}$, denoted by $\partial_\mathrm{C} \phi$, is defined as
	\begin{align*}
		&\partial_\mathrm{C} \phi(x):=\\
		&\operatorname{conv}\left\{\left.\lim _{i \rightarrow \infty} D\phi(x_i) \right| x_i \rightarrow x, \; \phi \text { is differentiable at } x_i\right\},
	\end{align*}
	where $D\phi(x_i) \in \mathbb{R}^{n_\phi\times n}$ is the classical Jacobian at points $x_i$ where $\phi$ is differentiable. The notation $n_\phi$ refers to the image dimension of $\phi$. 
\end{definition}
Note that for a locally Lipschitz function, the generalized Clarke derivative is nonempty, compact, convex, and may be set-valued \cite{clarke1976inverse}. It is the convex hull of the limits of classical Jacobian sequences at nearby points where the function is differentiable. If a function is differentiable at a point, the generalized Clarke derivative reduces to the classical Jacobian at that point.

\section{Safe controller synthesis}\label{section_synthesis}

For PWA systems, piecewise linear feedback controllers are often designed in the literature \cite{lazar2006stabilizing}. The PWA system \eqref{pwa_form} controlled by a piecewise linear backup controller can therefore be written as
\begin{equation}\label{closed_loop_pwa}
	\dot x = f_\mathrm{PWA,b} (x) := D_i x + d_i,\;\mathrm{for} \; x \in \mathcal{R}_i.
\end{equation}
In \eqref{closed_loop_pwa}, $f_\mathrm{PWA,b}:\mathbb{R}^n \to \mathbb{R}^n$ is a continuous PWA function, and $\{\mathcal{R}_i\}^r_{i=1}$ constitutes a strict polyhedral partition of the state space. Since \eqref{closed_loop_pwa} has a global Lipschitz constant $L= \max_{i\in \mathbb{N}^+_r} ||D_i||$, we can denote the unique solution with any initial condition  $x(0) = x_0 \in  \mathbb{R}^n$ by $x(\tau) = \phi_\mathrm{b}(x_0,\tau)$, $\forall \tau\in [0,\infty)$. Furthermore, according to \cite[Theorem 3.4]{khalil2002nonlinear}, for any two initial states $x_0$ and $y_0$, the solutions of \eqref{closed_loop_pwa} satisfy $||\phi_\mathrm{b}(x_0,\tau)- \phi_\mathrm{b}(y_0,\tau)|| \leq ||x_0 - y_0|| \mathrm{exp} (L \tau),\; \forall \tau>0$. This proves the Lipschitz continuous dependence of $\phi_\mathrm{b}(x_0,\tau)$ on the initial state $x_0$ in $\mathbb{R}^n$ for any finite $\tau$.

When $\phi_\mathrm{b}$, $h_\mathrm{b}$, or $h_X$ are not necessarily differentiable, we extend the smooth PSF formulation in \eqref{safety filter} by using their generalized Clarke derivatives:

\begin{subequations}\label{safety filter pwa}
	\begin{align}
		&\text{All-elements PSF:} \nonumber\\
		&	\pi_\mathrm{safe} (x)\nonumber\\
	=&\underset{u \in U}{\operatorname{argmin}}\left\|u-\pi_{\mathrm{r}}(x)\right\|^2\\	
		&\mathrm{s.t.}\; \partial h_X Q f_\mathrm{PWA}(x,u)  \geq-\alpha\left(h_X\left(\phi_{\mathrm{b}}( x,\tau)\right)\right),  \nonumber\\
		&\quad \;\;\forall \partial h_X \!\in \!\partial_\mathrm{C} h_X(\phi_{\mathrm{b}}( x,\tau)),\;\forall  Q \!\in\! \partial_\mathrm{C}\phi_\mathrm{b}(x,\tau),\; \forall \tau \!\in\! [0, T] \label{b} \\
		&\quad \;\;\partial h_{\mathrm{b}} Q  f_\mathrm{PWA}(x,u)   \geq-\alpha_{\mathrm{b}}\left(h_{\mathrm{b}}\left(\phi_{\mathrm{b}}( x,T)\right)\right),\nonumber\\
		&\quad \;\;\forall \partial h_{\mathrm{b}} \in \partial_\mathrm{C}  h_{\mathrm{b}}(\phi_{\mathrm{b}}( x,T)),\; \forall  Q \in \partial_\mathrm{C}\phi_\mathrm{b}(x,T), \label{c}
	\end{align}
\end{subequations}
where $ \partial_\mathrm{C}\phi_\mathrm{b}(x,\tau)$ is called the \emph{Clarke sensitivity}. Different from \eqref{safety filter} in which the CBF constraints are uniquely determined by the unique sensitivity matrix $\frac{\partial \phi_{\mathrm{b}}}{\partial x}$ and the unique derivatives $D h_X$ and $D h_\mathrm{b}$, we require that the CBF constraints should be satisfied for \emph{all} elements contained in the sets $\partial_\mathrm{C} h_X$, $\partial_\mathrm{C} h_\mathrm{b}$, and $\partial_\mathrm{C}\phi_\mathrm{b}$. 

\subsection{Performance analysis}
Before deriving a computationally tractable way to compute the Clarke sensitivity and to solve \eqref{safety filter pwa}, we first analyze the safety performance of $\pi_\mathrm{safe}$.

\begin{assumption}\label{A1}
	\noindent (i) $\pi_\mathrm{b}(x) \in U$, for all $x \in X$;
	
	\noindent (ii) $\partial h_\mathrm{b}f_\mathrm{PWA,b}(x) \geq -\alpha_\mathrm{b}(h_\mathrm{b}(x))$, for all $x \in S_\mathrm{b}:=\{x \in \mathbb{R}^n | h_\mathrm{b}(x) \geq 0\} \subseteq X$ and for all $\partial h_\mathrm{b} \in \partial_\mathrm{C} h_\mathrm{b}(x)$.
\end{assumption}

Assumption \ref{A1} is reasonable and aligns with similar assumptions commonly found in the literature on predictive CBFs \cite{chen2021backup,molnar2023safety} and MPC frameworks that incorporate CBFs or control Lyapunov functions for designing terminal constraints \cite{gurriet2020scalable}. The key distinction in our assumption is that it requires the CBF constraint to hold for all elements in the Clarke generalized gradient of $h_b$, rather than relying solely on the unique classical gradient as typically assumed in existing works. If $h_\mathrm{b}$ is smooth, then Assumption \ref{A1} becomes equivalent to other assumptions in the literature, such as \cite[Assumption 1]{molnar2023safety}.

\begin{theorem}[Feasibility]\label{Theorem_feasibility}
	Consider the continuous PWA system \eqref{pwa_form}, the piecewise linear backup controller $\pi_\mathrm{b}$, the Lipschitz continuous $h_X$ and $h_\mathrm{b}$, and the proposed PSF \eqref{safety filter pwa}. Suppose that Assumption \ref{A1} holds and that $X$ is compact. Then, there exist $\alpha,\;\alpha_{\mathrm{b}} \in \mathcal{K}_{\infty}$ such that problem \eqref{safety filter pwa} is feasible for any $x \in \Phi_{\text {const}}\left(S_{\mathrm{b}}, T\right)$ and any $T\geq0$.
\end{theorem}

\begin{proof}
	We will prove that the backup solution $\pi_\mathrm{b}(x)$ is a feasible solution to \eqref{safety filter pwa}. To do so, for any $\tau \in [0,T]$, we need to consider several cases depending on the differentiability of the functions $\phi_\mathrm{b}$, $h_\mathrm{b}$, and $h_X$. For clarity and without loss of generality, we consider the following two cases: (i) $\phi_{\mathrm{b}}(x, \tau)$ is differentiable, and (ii) $\phi_{\mathrm{b}}(x, \tau)$ is not differentiable. The same reasoning applies analogously when considering the differentiability or non-differentiability of $h_\mathrm{b}$ and $h_X$, and thus we omit the cases for these two functions for brevity. 
	
	(i) If $\phi_{\mathrm{b}}(x, \tau)$ is differentiable w.r.t. $x$ at $x = \bar x$, then $\partial_\mathrm{C}\phi_\mathrm{b}(x,\tau)$ reduces to the classical sensitivity $\frac{\partial \phi_{\mathrm{b}}(x, \tau)}{\partial x}$ and therefore the feasibility of $\pi_\mathrm{b}(x)$ follows from Lemma \ref{lemma_safety}.
	
	(ii) If $\phi_{\mathrm{b}}(x, \tau)$ is not differentiable w.r.t. $x$ at $x = \bar x$, noting Definition \ref{Clarke definition}, consider any sequence $\{x_i\}^{\infty}_{i=0}$ converging to $\bar x$ and satisfying that $\phi_{\mathrm{b}}(x, \tau)$ is differentiable w.r.t. $x$ at $x_i$. It follows from the analysis in the first case that by letting $ u_i = \pi_\mathrm{b} (x_i)$, we have that
	\begin{equation*}
		D\! h_X\!\left(\phi_{\mathrm{b}}( x_i,\tau)\right) \frac{\partial \phi_{\mathrm{b}}(x_i, \tau)}{\partial x_i} f_\mathrm{PWA}(x_i,u_i) \! \geq \!-\!\alpha\!\left(h_X\left(\phi_{\mathrm{b}}( x_i,\tau)\right)\right)
	\end{equation*}
	holds. By taking the limit $i \to \infty$ and noting that $\lim_{i \to \infty} \frac{\partial \phi_{\mathrm{b}}(x_i, \tau)}{\partial x_i} \in \partial_\mathrm{C} \phi_{\mathrm{b}}(\bar x, \tau )$ by Definition \ref{Clarke definition}, we get the feasibility of \eqref{b}. An analogous argument can be applied to prove the feasibility of \eqref{c}. The proof of Theorem \ref{Theorem_feasibility} is completed.
\end{proof}

To establish the forward invariance of $\Phi_{\text{const}}(S_{\mathrm{b}}, T)$, i.e., the infinite-horizon safety of $\pi_\mathrm{safe}$, we introduce the following key lemma based on the comparison theorem \cite[Lemma 3.4]{khalil2002nonlinear}. This lemma is applicable to invariance analysis for general nonlinear systems involving any non-smooth CBF.

\begin{lemma}\label{lemma_invariance}
	Consider the system \eqref{pwa}. Suppose that there exist a Lipschitz continuous controller $\pi(\cdot) : \mathbb{R}^n \to \mathbb{R}^m$, a Lipschitz continuous $h(\cdot): \mathbb{R}^n \to \mathbb{R}$, and a Lipschitz continuous function $ \alpha \in \mathcal{K}_{\infty}$ such that
	\begin{align}\label{cbf_condition}
		&\partial h(x) f(x,\pi(x)) \geq - \alpha (h(x)),\nonumber\\
		&\forall x \in \{x \in \mathbb{R}^n| h(x) \geq 0\} \text{ and } \forall \partial h(x) \in \partial_\mathrm{C} h(x).
	\end{align}
	As a result, if $h(x_0) \geq 0$, we have $h(\phi(x_0,t)) \geq 0,\;\forall t\geq 0$, where $\phi(x_0,t)$ is the solution of the system \eqref{pwa} controlled by $\pi$, starting from $x_0$.
\end{lemma}

\begin{proof}
	 See Appendix \ref{appendix_lemma_invariance}.
\end{proof}

\begin{remark}
	\textcolor{blue}{The result in Lemma \ref{lemma_invariance} establishes the forward invariance of $\{x \in \mathbb{R}^n| h(x) \geq 0\}$. This result extends safety guarantees from continuously differentiable CBFs to CBFs that are merely Lipschitz continuous. More importantly, in contrast to the non-smooth CBF analysis in \cite{glotfelter2017nonsmooth}, we do not require $\dot{h}(t) \geq -\alpha(h(t))$ for almost every $t$. In particular, Lemma \ref{lemma_invariance} broadens the safety theory of non-smooth CBFs to the case where the set of times at which the trajectory encounters non-smoothness has non-zero measure. This can occur for PWA systems, which may operate on the shared boundary of two adjacent polytopes for some periods, as discussed in detail in the next section. It can also occur for smooth dynamics that operate in a region where $h$ is not differentiable.}
\end{remark}

With the safety guarantees of non-smooth CBFs given in Lemma \ref{lemma_invariance}, we are ready to present the main result on the safety of $\pi_\mathrm{safe}$.

\begin{theorem}[Forward invariance]\label{theoreom_forward}
	Consider the continuous PWA system \eqref{pwa_form}, the piecewise linear backup controller $\pi_\mathrm{b}$, the Lipschitz continuous backup CBF $h_\mathrm{b}$, and the proposed all-elements PSF \eqref{safety filter pwa}. If Assumption \ref{A1} holds and the policy $\pi_\mathrm{safe}$ is locally Lipschitz continuous, then $\pi_\mathrm{safe}$ renders $\Phi_{\text {const}}\left(S_{\mathrm{b}}, T\right)$ forward invariant.
\end{theorem}

\begin{proof}
	Given $T\geq 0$, like \cite{chen2021backup}, we consider the CBF candidate in \eqref{CBF_candidate}, which is Lipschitz continuous and satisfies $\Phi_{\text {const}}\left(S_{\mathrm{b}}, T\right) = \{x\in\mathbb{R}^n| h(x) \geq 0\}$ according to \cite[Lemma 2]{chen2021backup}.
	
	For each $x \in \Phi_{\text {const}}\left(S_{\mathrm{b}}, T\right)$, based on \eqref{CBF_candidate}, we can distinguish two cases: (i) $h(x) = h_X\left(\phi_\mathrm{b}(x, \tau')\right)$ with a certain $\tau' \in [0,T]$ and (ii) $h(x) = h_\mathrm{b}\left(\phi_\mathrm{b}(x, T)\right)$. 
	
	For the first case, we define the composite function $h_\phi = h_X \circ \phi_\mathrm{b}$. Given $\tau' \in [0,T]$, we have the following expression for the Clarke derivative of the composite function \cite[Theorem 2.3.9]{clarke1990optimization}:
	\begin{align*}
		&\partial_\mathrm{C} h_\phi(x,\tau')\\
        \subset& \mathrm{conv}\left\{\partial h_X \;Q,\; \partial h_X \in \partial_\mathrm{C} h_X(\phi_{\mathrm{b}}(x,\tau')),\; Q \in \partial_\mathrm{C}\phi_\mathrm{b}(x,\tau') \right\}.
	\end{align*}
	Then, we can apply the chain rule to obtain the Clarke derivative of the composite function $h$, that is, for any $\partial h \in \partial_\mathrm{C} h(x)$,
	\begin{align}\label{deviation_CBF}
		&\partial h f_\mathrm{PWA}(x,\pi_\mathrm{safe}(x)) + \alpha (h(x)) \nonumber\\
		=&  \partial h_\phi f_\mathrm{PWA}(x,\pi_\mathrm{safe}(x)) + \alpha (h_X(x)),\;\exists \partial h_\phi \in \partial_\mathrm{C} h_\phi(x,\tau')\nonumber\\
		=& \partial h_X Q  f_\mathrm{PWA}(x,\pi_\mathrm{safe}(x)) + \alpha\left(h_X\left(\phi_{\mathrm{b}}( x,\tau')\right)\right),\nonumber\\
		& \quad\quad\quad\quad \exists \partial h_X \in \partial_\mathrm{C} h_X(\phi_\mathrm{b}(x, \tau')),
		 Q \in \partial_\mathrm{C}\phi_\mathrm{b}(x,\tau')\nonumber\\
		\geq &0,
	\end{align}
	where the second equality is obtained by applying the chain rule to $\partial_\mathrm{C} h_\phi$, and the last line follows from the fact that  $\pi_\mathrm{safe}$ satisfies the constraint \eqref{b}. 
	
	An analogous argument to that in \eqref{deviation_CBF} can be applied for the second case, yielding $\partial h f_\mathrm{PWA}(x,\pi_\mathrm{safe}(x))+ \alpha_\mathrm{b} (h(x)) \geq 0$ for any $\partial h  \in \partial_\mathrm{C} h(x)$.
	
	As a consequence, the CBF candidate $h$ satisfies \eqref{cbf_condition} in Lemma \ref{lemma_invariance}. If $\pi_\mathrm{safe}$ is locally Lipschitz continuous, applying Lemma \ref{lemma_invariance} yields $h(\phi_\mathrm{safe}(x_0,t)) \geq 0,\;\forall t\geq 0$, if $x_0 \in \Phi_{\text {const}}\left(S_{\mathrm{b}}, T\right)$. Here, $\phi_\mathrm{safe}(x_0,t)$ is the solution of the system \eqref{pwa_form} starting from $x_0$ controlled by $\pi_\mathrm{safe}$. This proves the forward invariance of $\Phi_{\text {const}}\left(S_{\mathrm{b}}, T\right)$ under $\pi_\mathrm{safe}$.
\end{proof}

Theorems \ref{Theorem_feasibility} and \ref{theoreom_forward} state that when there are non-smooth CBF constraints, the all-elements PSF can ensure that the system always satisfies the constraints, provided that $\pi_\mathrm{safe}$ is locally Lipschitz continuous. Analyzing the conditions under which $\pi_\mathrm{safe}$ is Lipschitz continuous is left for future work.

\textcolor{blue}{Assumption \ref{A1} can be further relaxed by replacing the safe set $S_\mathrm{b}$ with the subset $S_\mathrm{b}\cap \mathcal{O}$, i.e., assuming that the CBF inequality holds over $S_\mathrm{b}\cap \mathcal{O}$, where $\mathcal{O}$ is any forward invariant set for the backup closed-loop system \eqref{closed_loop_pwa}. As a result, since both the reachable set computation $\Phi_\mathrm{const}(\cdot,T)$ and set intersection preserve invariance \cite{blanchini2008set}, it is concluded that $\pi_\mathrm{safe}$ will render $\Phi_\mathrm{const}\left(S_\mathrm{b}\cap \mathcal{O},T\right)$ forward invariant. The above argument is summarized in the following corollary.}
\textcolor{blue}{\begin{corollary}[Corollary of Theorem 2]\label{corollary_forward}
	Consider the continuous PWA system \eqref{pwa_form}, the piecewise linear backup controller $\pi_\mathrm{b}$, the Lipschitz continuous backup CBF $h_\mathrm{b}$, and the proposed all-elements PSF \eqref{safety filter pwa}. If (a) Assumption \ref{A1}(i) holds, (b) $\partial h_\mathrm{b}f_\mathrm{PWA,b}(x) \geq -\alpha_\mathrm{b}(h_\mathrm{b}(x))$, for all $x \in S_\mathrm{b}\cap  \mathcal{O}$ and for all $\partial h_\mathrm{b} \in \partial_\mathrm{C} h_\mathrm{b}(x)$, where $\mathcal{O}$ is forward invariant for \eqref{closed_loop_pwa}, and (c) the policy $\pi_\mathrm{safe}$ is locally Lipschitz continuous, then $\pi_\mathrm{safe}$ renders $\Phi_{\text {const}}\left(S_{\mathrm{b}}\cap\mathcal{O}, T\right)$ forward invariant.
\end{corollary}
}

\section{Non-smooth analysis of PWA systems}\label{section_nonsmooth}

% Furthermore, according to \cite[Theorem 3.4]{khalil2002nonlinear}, for any two initial states $x_0$ and $y_0$, the solutions of \eqref{closed_loop_pwa} satisfy $||\phi_\mathrm{b}(x_0,\tau)- \phi_\mathrm{b}(y_0,\tau)|| \leq ||x_0 - y_0|| \mathrm{exp} (L \tau),\; \forall \tau>0$. This proves the Lipschitz continuous dependence of $\phi_\mathrm{b}(x_0,\tau)$ on the initial state $x_0$ in $\mathbb{R}^n$ for any finite $\tau$. Additionally, according to Rademacher's Theorem \cite{evans2018measure}, $\phi_\mathrm{b}(x_0,\tau)$ is differentiable w.r.t. $x_0$ almost everywhere. 

%Below, we provide a more detailed analysis of conditions under which the generalized Clarke derivative can be used to describe the sensitivity of $\phi_\mathrm{b}$ w.r.t. $x_0$.

In the previous section, we have shown sufficient conditions to ensure formal guarantees for feasibility and safety for the proposed all-element PSF \eqref{safety filter pwa}. However, To construct the proposed PSF, we need to compute the generalized Clarke derivative for $\phi_\mathrm{b}$ (Clarke sensitivity) and the generalized Clarke derivatives $\partial_\mathrm{C} h_X$ and $\partial_\mathrm{C} h_\mathrm{b}$. Computing $\partial_\mathrm{C} h_X$ and $\partial_\mathrm{C} h_\mathrm{b}$ is generally feasible, since their explicit forms are usually known. Finding a generalized Clarke derivative for $\phi_\mathrm{b}$ requires computing classical Jacobians for all converging sequences $x_i$ where $\phi_\mathrm{b}$ is differentiable, and it is generally intractable. This section addresses this challenge for the considered PWA system.

The following definition is useful for characterizing the behavior of PWA systems:
\begin{definition}[Time-stamped switching sequence]
	Consider the PWA system \eqref{closed_loop_pwa}. Given the initial state $x_0$, the \emph{time-stamped switching sequence} of the solution $\phi_\mathrm{b}(x_0,\tau)$, denoted as 
	\begin{align}\label{switching}
		\Gamma(x_0) = &\{(\tau_0,I_0),\;(\tau_1,I_1),\;...,\;(\tau_M,I_M)\},\nonumber\\
		&\text{ with }I_k \subseteq \mathbb{N}^+_r,\;\forall k\in \mathbb{N}_M,
	\end{align}
	is the ordered sequence of time instants and the indices of regions such that 
	\begin{itemize}
		\item $\tau_0 =0$ and $x_0 \in \mathrm{clo}(\mathcal{R}_{I_0})$;
		\item $\phi_\mathrm{b}(x_0,\tau) \in \mathrm{clo}(\mathcal{R}_{i}),\;\forall i \in I_k\text{ and } \forall \tau \in [\tau_k, \tau_{k+1}) $;
		\item $I_k \ne I_{k+1},\;\forall k \in \mathbb{N}_{M-1}$ and $\tau_0 < \tau_1<...<\tau_M$.
	\end{itemize}
	In \eqref{switching}, $r$ is the number of polyhedral regions of \eqref{closed_loop_pwa} and $M$ is the number of switching times. 
\end{definition}

The switching sequence can be infinite, i.e., $M = \infty$, meaning that the system exhibits zeno behavior if $\tau_M$ is finite or chattering if $\tau_M = \infty$. Besides, each $I_k$ can contain more than one element. In this case, the system operates on the shared boundary of multiple regions of \eqref{closed_loop_pwa} during $[\tau_k,\tau_{k+1})$. If $I_k$ consists of multiple elements, then we define $\mathcal{R}_{I_k} := \bigcup_{i \in I_k} \mathcal{R}_i$.

\subsection{Aumann sensitivity for PWA systems}
In this subsection, we propose an alternative and computable generalized sensitivity $Q_\mathrm{A}(x_0,\tau)$ using the Aumann integral \cite{aumann1965integrals}. In subsequent subsections, we will show that, under certain assumptions, this formulation is equivalent to the Clarke sensitivity.

%We can write the solution to \eqref{closed_loop_pwa} as 
%\begin{equation}\label{key}
%	\phi_\mathrm{b}(x_0,\tau) = x_0 + \int_{0}^\tau f_\mathrm{PWA,b}(\phi_\mathrm{b}(x_0,s)) d s. 
%\end{equation}

\begin{definition}[Aumann sensitivity]\label{Aumann}
	Consider the PWA system \eqref{closed_loop_pwa} and its solution $\phi_\mathrm{b}$. The \emph{Aumann sensitivity} for any initial state $x_0$ at time $\tau$, denoted by $Q_\mathrm{A}(x_0,\tau)$, is implicitly defined by
	\begin{align}\label{integral}
		Q_\mathrm{A}(x_0,\tau)&: =  \left\{Q(x_0,\tau) \left| Q(x_0,\tau)=I_n\right. \right. \nonumber\\
        &+  \int_{0}^\tau \partial f_\mathrm{PWA,b} (\phi_{\mathrm{b}}( x_0,s)) Q(x_0,s)d s,\nonumber\\
		&\left.\partial f_\mathrm{PWA,b}(x)\in \{D_i | i \text{ s.t. }x \in \mathrm{clo}(\mathcal{R}_i)\}\right\}.
	\end{align}
\end{definition}

The Aumann integral is introduced to handle cases where the integration is carried out on all possible values of the integrand. As a result, the Aumann sensitivity naturally becomes a set-valued function. When this sensitivity takes a single value at a point, it coincides with the classical Jacobian of $\phi_\mathrm{b}$ at that point. 

Given the initial state $ x_0$, denote by $I(x_0)$ the Cartesian product of the switching sequence $\{I_0, I_1,...,I_M\}$. With this definition, the Aumann sensitivity is expressed as:
\begin{align}\label{generalized sensitivity computation}
	&Q_\mathrm{A}(x_0,\tau) \nonumber\\
    \!=\!& \left\{Q(x_0,\tau) \left| Q(x_0,\tau)\! =\! I \!+  \!\sum_{k=0}^{\bar{k}} \int_{\tau_k}^{\min\{\tau,\tau_{k+1}\}} D_{i_k} Q(x_0,s) d s,\begin{array}{l}
		\;\\\;
	\end{array}\right.\right. \nonumber\\
	&\left.\begin{array}{l}
		\;\\\;\\\;
	\end{array}(i_0,i_1,...,i_M) \in I(x_0)\right\},\tau \!\in\! [\tau_{\bar{k}}, \tau_{\bar{k}+1}),\bar{k} \!\leq\! M\!-\!1.
\end{align}
Each $Q$ is single-valued and can be obtained by the standard sensitivity computation \cite[Chapter 3.3]{khalil2002nonlinear}. The cardinality of $Q_\mathrm{A}(x_0,\tau)$ is less than or equal to $\Pi_{k \in \mathbb{N}_M} |I_k|$. 

\subsection{Properties of the Aumann sensitivity}

The ultimate objective of this section is to show the equivalence between the Aumann sensitivity and the Clarke sensitivity. Before doing that, we first study the properties of the Aumann sensitivity.

Without loss of generality, suppose that the closure of each $\mathcal{R}_i$ admits the hyperplane representation $\{x\in \mathbb{R}^n | H_i x + h_i \leq 0\}$, where $H_i \in \mathbb{R}^{m_i \times n}$, $h_i \in \mathbb{R}^{m_i}$, and $m_i$ is the number of inequalities that constitute $\mathcal{R}_i$. Denote by 
\begin{equation*}
	\Gamma_I: = \{x\in \mathbb{R}^n | \bar{g}_{I} (x):= g_{I} x+b_{I}= 0\}
\end{equation*}
the common boundary of two or more adjacent regions $\mathcal{R}_I$, where $I$ is a set of indices of the adjacent regions. Here, $g_{I} \in \mathbb{R}^{n_{g_I} \times n}$, $b_{I} \in \mathbb{R}^{n_{g_I}}$, and $\bar{g}_{I}(\cdot):\mathbb{R}^n \to \mathbb{R}^{n_{g_I}}$.

From \eqref{integral} we know that $Q_\mathrm{A}$ could take multiple values if the solution $\phi_\mathrm{b}$ stays on the common boundary of some regions of \eqref{closed_loop_pwa} during a period. This is because during that period, the integral term can take multiple values. This motivates us to check if there is a non-empty set on the boundary such that the solution $\phi_\mathrm{b}$ belongs to this set during a certain period. We consider the following \emph{critical set} (could be empty):
\begin{align}\label{crtical set}
	\mathcal{C}&:= \bigcup_{I \in\mathcal{I}} \mathcal{C}_I, \text{ where } \mathcal{C}_I= \bigcap_{i\in I}\mathrm{clo}(\mathcal{R}_i)\nonumber \\ \mathcal{I}&: =\{I \subseteq \mathbb{N}^+_r | \{\mathcal{R}_i\}_{i\in I} \text{ are adjacent}, \bar{g}^{(q)}_{I}=0,\;\forall q \in \mathbb{N}_n,\nonumber\\
	&\quad \quad \quad \quad \quad \quad \text{ and }\exists i,j\in I \text{ s.t. } D_i \ne D_j  \}.
\end{align}
In \eqref{crtical set}, $\bar{g}^{(q)}_{I}$ is the $q$-th order element-wise time derivative of $\bar{g}_{I}$ along the PWA system \eqref{closed_loop_pwa}. Because the first $q$ order time derivatives of $\bar{g}_{I}$ are zero, $\bar{g}_{I}$ remains constant according to the Cayley–Hamilton Theorem \cite{decell1965application}. Therefore, we can interpret $\mathcal{C}$ as the set where the trajectory is possible to stay on a boundary $\bar{g}_{I} (x) =0$ of some polyhedra $\mathcal{R}_i,\; i\in I $ for some periods. 

For any initial state $x_0$, define the \emph{forward reachable set} as $\Phi_\mathrm{for}(x_0): = \bigcup_{\tau\in [0,\infty)} \{\phi_\mathrm{b}(x_0,\tau)\} $. Similarly, for any set $S$ of states, define the \emph{backward reachable set} $\Phi_\mathrm{back}(S) :=\{x \in \mathbb{R}^n | \exists \tau\in [0,\infty) \text{ s.t. } \phi_\mathrm{b}(x,\tau) \in S \}$. We have the following result on the relation among $\mathcal{C}$, $Q_\mathrm{A}$, and the classical Jacobian (if it exists).

\begin{proposition}\label{proposition1}
	The following statements are equivalent:
	
	(1) The classical Jacobian $\frac{\partial \phi_\mathrm{b}(x_0,\tau)}{\partial x_0}$ of $\phi_\mathrm{b}(x_0,\tau)$ exists at $x_0 = \bar x_0$ for all $\tau\geq 0$.
	
	(2) The Aumann sensitivity $Q_\mathrm{A}(x_0,\tau)$ is single-valued at $x_0 = \bar x_0$ for all $\tau\geq 0$.
	
	(3) $\Phi_\mathrm{for}(\bar x_0) \cap \mathcal{C} = \emptyset$.
	
	(4) $\bar x_0 \notin \Phi_\mathrm{back}(\mathcal{C})$.
	
	As a consequence, if $\mathcal{C}$ is empty, the Aumann sensitivity $Q_\mathrm{A}$ defined in \eqref{integral} is single-valued for any $x_0$, which implies the global existence of the classical Jacobian $\frac{\partial \phi_\mathrm{b}(x_0,\tau)}{\partial x_0}$.
\end{proposition}

\begin{proof}
	See Appendix \ref{appendix_proposition1}.
\end{proof}

The critical set $\mathcal{C}$ is a union of some polyhedra $\mathcal{C}_{I},\; I \in \mathcal{I} $, with the dimension strictly lower than $n$. By using the Lie derivative for \eqref{crtical set}, each $\mathcal{C}_{I},\;I \in \mathcal{I}$ can be expressed as
\begin{equation}\label{SI}
	\mathcal{C}_I = \left\{ 
	x \in \mathbb{R}^n \,\middle|\, 
	\begin{aligned}
	\left.{\begin{aligned}
		&g_I x + b_I = 0 \\
		&g_I(D_i x + d_i) = 0 \\
		&g_I D_i (D_i x + d_i) = 0 \\
		&\vdots \\
		&g_I D_i^{n-1}(D_i x + d_i) = 0 \\
		&H_i x + h_i \leq 0 
	\end{aligned}} \right\} \forall i \in I \\
 \text{ and }\exists i, j \in I \text{ s.t. } D_i \ne D_j
	\end{aligned}
	\right\}
\end{equation}
The non-emptiness of each $\mathcal{C}_{I}$ can be determined by checking the feasibility of a linear program. %Furthermore, note that due to the continuity of the PWA system, $D_i x + d_i = D_j x + d_j,\;\forall x\in \Gamma_I$. This means that some equalities in \eqref{SI} are redundant. %Therefore, the linear program contains $(n+1)n_{g_I}$ equality constraints and $\sum_{i\in I} m_i$ inequality constraints.

\begin{remark}
	Proposition \ref{proposition1} implies that the sensitivity $Q_\mathrm{A}$ is uniquely defined as long as the state does not remain on a cell boundary for any time interval, but always “passes through”. An important fact indicated by Proposition \ref{proposition1} is that the non-existence of $\frac{\partial \phi_\mathrm{b}\left(x_0, \tau \right)}{\partial x_0}$ does not necessarily occur on the boundary of the PWA system; rather, it can also take place in the interior of any polyhedron $\mathcal{R}_i$. The following example illustrates the set $\mathcal{C}$ and also the set of initial states at which $\frac{\partial \phi_\mathrm{b}\left(x_0, \tau \right)}{\partial x_0}$ fails to exist.
\end{remark}

\textbf{Example 1.} Consider a two-dimensional PWA system illustrated in Fig. \ref{example1}. The system is partitioned by three lines: $x_1=0$, $x_1=2$, and $x_2=0$, resulting in six polyhedra. The critical set $\mathcal{C}$, according to \eqref{SI}, consists of the positive segment of the $x_1$-axis. For every initial point along the red and blue lines, when the trajectory $\phi_\mathrm{b}$ reaches the right-half plane, it lacks a well-defined Jacobian. To prove this, we first note that for any initial state $\bar x_0$ starting from the red line $(\mathcal{C})$, the solution will always remain in $\mathcal{C}$. Therefore, $\Phi_\mathrm{for}(\bar x_0) \cap \mathcal{C} \ne \emptyset$, and $\frac{\partial \phi_\mathrm{b}(x_0,\tau)}{\partial x_0}$ does not exist, according to Proposition \ref{proposition1}. Then, we analyze the initial states in the first quadrant. By solving the affine system in this quadrant, we get the evolution: $x_1(\tau) = x_1(0)+\tau,\;x_2(\tau) = (x_1(0)+x_2(0)+1)e^\tau-\tau-x_1(0)-1$, for $\tau \geq 0$ such that $x_1(\tau)\leq 0$ and $x_2(\tau)\geq 0$. If $\Phi_\mathrm{for}(\bar x_0) \cap \mathcal{C} \ne \emptyset$, at a certain time instant $\tau_{k_1}$ it must hold that $x_1(\tau_{k_1}) = x_2(\tau_{k_1}) = 0$. Otherwise, the trajectory can only converge to the red line asymptotically. From $x_1(\tau_{k_1}) = x_2(\tau_{k_1}) = 0$, we can derive the relation:
\begin{equation}\label{eq_example}
	x_2(0) = e^{x_1(0)}-x_1(0) -1.
\end{equation}
The blue line corresponds to \eqref{eq_example} in the first quadrant.

By simulating the system starting from the points near the blue line, we can find that even a small perturbation on the initial state will result in a large gap when the trajectory goes into the right-half plane. This validates that the classical Jacobian does not exist at the initial states on the blue curve after the critical time when the trajectory enters the red line. 

\begin{figure}\label{tra_figure}
	\centering
	\includegraphics[width=160pt,clip]{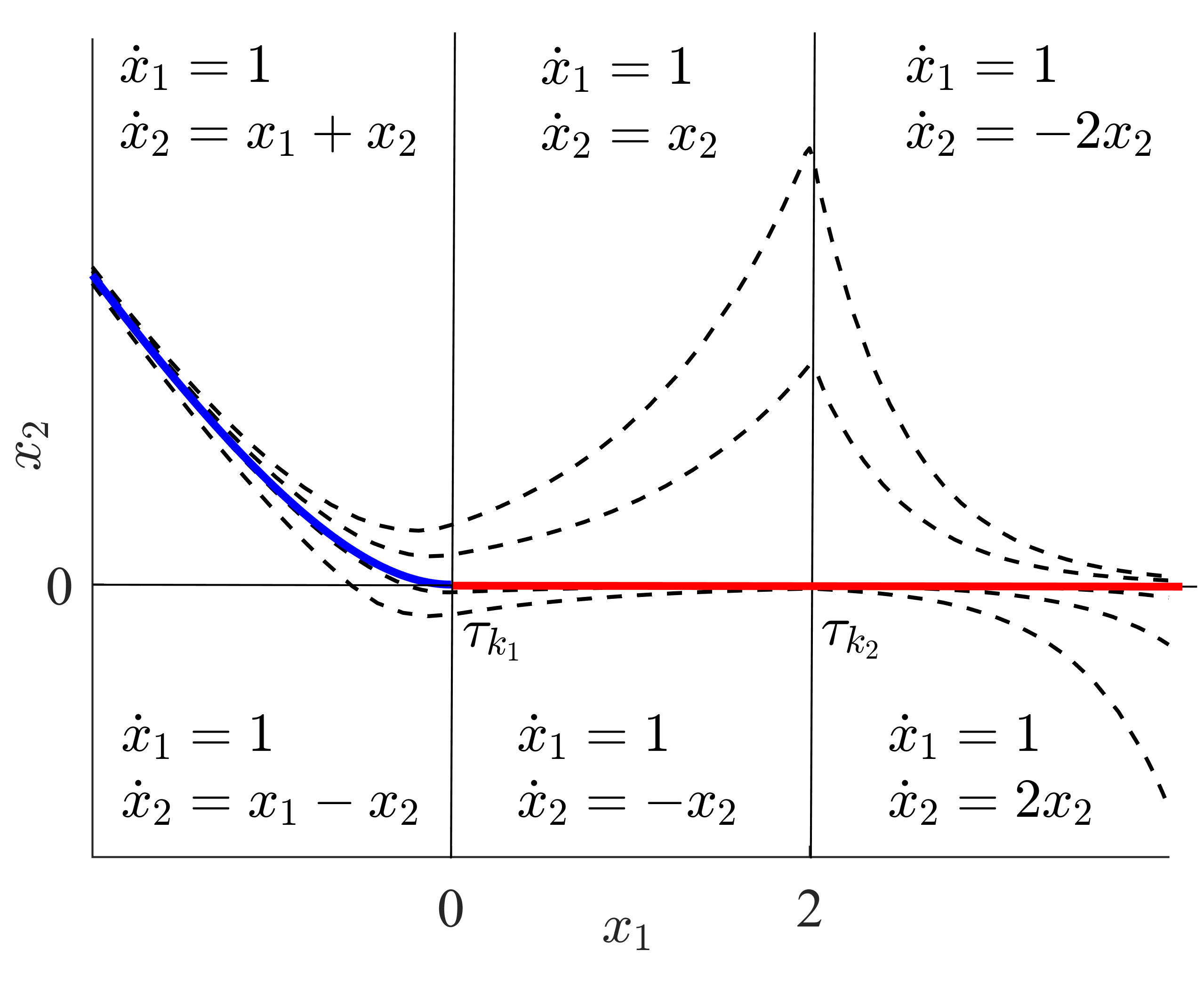}
	\caption{The evolution of the PWA system given in Example 1. The black solid lines refer to the boundaries of the polyhedra. The critical set $\mathcal{C}$ is the positive part of the $x_1$-axis, represented by the red line. The backward reachable set $\Phi_{\text {back }}(\mathcal{C})$ is the union of the red and blue lines. According to Proposition \ref{proposition1}, the solution $\phi_\mathrm{b}$ is not differentiable at all points on the red and blue lines. The black dashed curves represent some trajectories starting near the blue line. }
    \label{example1}
\end{figure}

The sensitivity branches at the time $\tau$ when the solution $\phi_{\mathrm{b}}(x_0, \tau)$ enters the critical set $\mathcal{C}$. Fig. \ref{example2} gives a geometric interpretation of $Q_\mathrm{A}( x_0,\tau)$. 

\begin{figure}\label{tree_figure}
	\centering
	\includegraphics[width=200pt,clip]{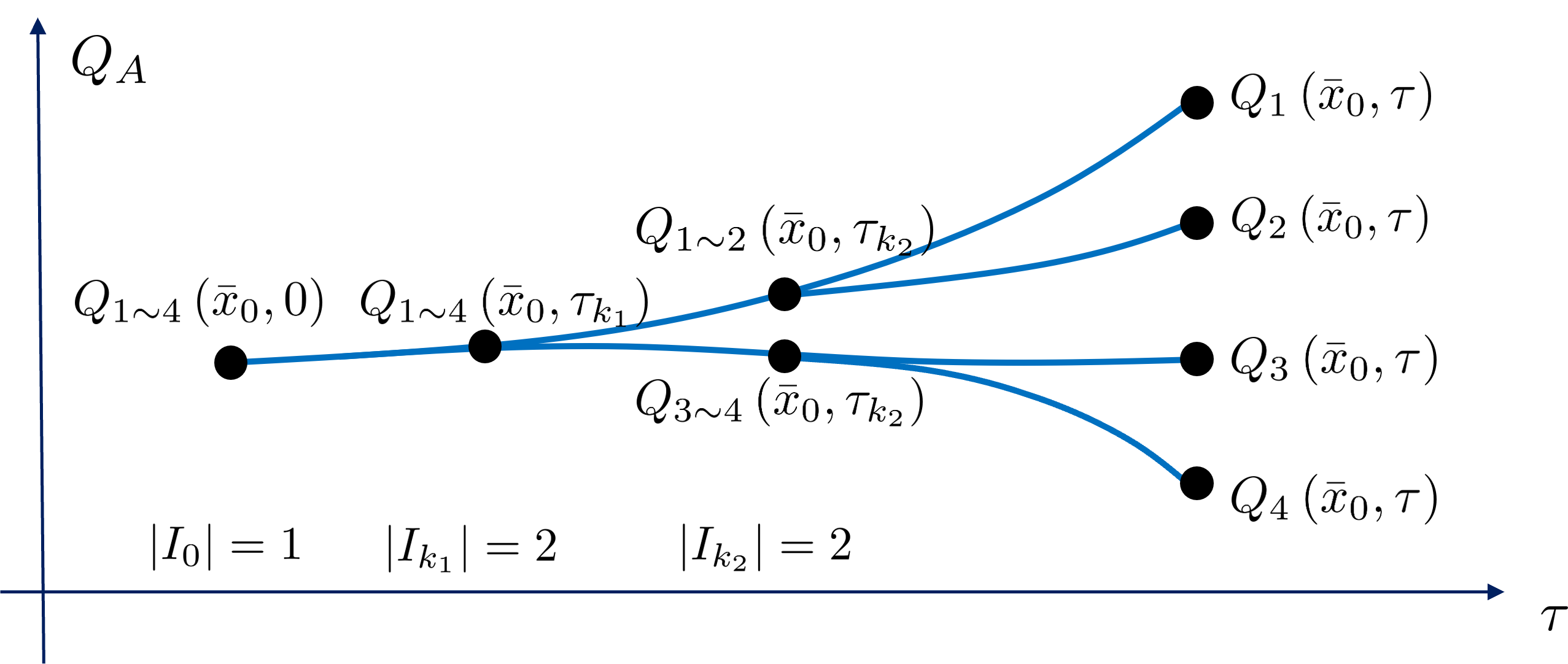}
	\caption{Evolution of the sensitivity function $Q_\mathrm{A}$ over time for any initial state at the blue line in Fig. \ref{example1}. The components $Q_1$ to $Q_4$ represent individual elements of $Q_\mathrm{A}$. The sensitivity branches when the solution $\phi_{\mathrm{b}}(\bar x_0, \tau)$ enters each critical set. In Example 1, this occurs at $\tau_{k_1}$ when the trajectory reaches $(0,0)$ and at $\tau_{k_2}$ when the trajectory reaches $(2,0)$. }
    \label{example2}
\end{figure}

\subsection{Equivalence between $\mathrm{conv}(Q_\mathrm{A})$ and $\partial_\mathrm{C} \phi_\mathrm{b}$ }

In this subsection, we show that $\mathrm{conv}(Q_\mathrm{A} (x_0,\tau))= \partial_\mathrm{C} \phi_\mathrm{b}(x_0,\tau),\;\forall x_0 \in \mathbb{R}^n$ and $\forall \tau \in [0,\infty)$. Toward this end, the following lemma and assumption are useful.

\begin{lemma}\label{invertible}
	Consider the PWA system \eqref{closed_loop_pwa} and its Aumann sensitivity $Q_\mathrm{A}$ defined in \eqref{integral}. For any $x_0 \in \mathbb{R}^n$, any $\tau \in [0,\infty)$, and any $Q(x_0,\tau) \in Q_\mathrm{A}(x_0,\tau)$, $Q(x_0,\tau)$ is invertible. 
\end{lemma}
\begin{proof}
	According to \eqref{generalized sensitivity computation}, any $Q(x_0,\tau)$ can be recursively computed by
	\begin{equation}\label{recursive}
		Q(x_0,t) = \left\{ \begin{gathered}
			{e^{{D_{{i_0}}}\tau}},\tau \in [{\tau_0},{\tau_1}) \hfill \\
			{e^{{D_{{i_1}}}(\tau - {\tau_1})}}Q(x_0,{\tau_1}),t \in [{\tau_1},{\tau_2}) \hfill \\
			... \hfill \\
			{e^{{D_{{i_{M - 1}}}}(\tau - {\tau_{M - 1}})}}Q(x_0,{\tau_{M - 1}}),\tau \in [{\tau_{M - 1}},{\tau_M}),
		\end{gathered}  \right.
	\end{equation}
	where $\{i_0,i_1,...,i_M\} \in I(x_0)$ and $I(x_0)$ is the Cartesian product of the switching sequence $\{I_0,I_1,...,I_M\}$ for given $x_0$.
	
	As ${e^{{D_{{i_0}}}\tau}}$ is always invertible, and $A B$ is invertible for any two invertible matrices $A$ and $B$, it is straightforward to prove by induction that $Q(x_0,\tau)$ is invertible for any $x_0 \in \mathbb{R}^n$ and $\tau \in [0,\infty)$.
\end{proof}

\begin{assumption}\label{assumption_invert}
	For any $x_0 \in \Phi_{\text {back }}(\mathcal{C})$, any $\tau \in [0,\infty)$, and any $Q(x_0,\tau) \in \mathrm{conv}(Q_\mathrm{A}(x_0,\tau))$, $Q(x_0,\tau)$ is invertible. 
\end{assumption}

\begin{remark}
	Since any $Q \in Q_\mathrm{A}$ is invertible according to Lemma \ref{invertible}, we further assume that any element in the convex hull of $Q_\mathrm{A}$ is invertible. Although convex combinations may not preserve invertibility, Assumption \ref{assumption_invert} is typically satisfied in practical scenarios such as adaptive cruise control with PWA models in \cite{corona2008adaptive} and the inverted pendulum system considered in \cite{he2023state} and the case study of the current paper. This assumption is crucial in the subsequent analysis on the equivalence between $\mathrm{conv}(Q_\mathrm{A})$ and $\partial_\mathrm{C} \phi_\mathrm{b}$, as it enables the application of the inverse function theorem \cite[Theorem 1]{clarke1976inverse} for the non-smooth solution $\phi_{\mathrm{b}}$. 
\end{remark}

The following lemma gives a sufficient LMI condition for Assumption \ref{assumption_invert} to hold.
{\color{blue}
\begin{lemma}\label{lemma_invert_2}
	For each $x_0 \in \Phi_{\text {back }}(\mathcal{C})$ and each $\tau \in [0,\infty)$, if there exists a $P\in \mathbb{R}^{n\times n}$ such that
    \begin{equation}\label{eq:lmi}
        P Q_l+Q_l^T P^T \succ 0,\; \forall Q_l \in Q_\mathrm{A}(x_0,\tau),\;l=1,2,..., K,
    \end{equation}
    then any $Q(x_0,\tau) \in \mathrm{conv}(Q_\mathrm{A}(x_0,\tau))$ is invertible. 
\end{lemma}
\begin{proof}
Since any $Q(x_0,\tau) \in \mathrm{conv}(Q_\mathrm{A}(x_0,\tau))$ can be represented as  $Q(x_0,\tau) = \sum_{l=1}^{K} \mu_l Q_l$ with $\sum_{l=1}^{K} \mu_l=1$ and $\mu_l\in[0,1]$, \eqref{eq:lmi} implies
 \begin{equation}\label{eq:lmi2}
        P Q(x_0,\tau) +Q^T(x_0,\tau) P^T \succ 0
\end{equation}
Then, we prove the statement by contradiction. Suppose $Q(x_0,\tau)$ is not invertible. Then there exists $v\neq 0$ such that $Q(x_0,\tau)v=0$. Multiplying the left-hand side of \eqref{eq:lmi2} from the left by $v^T$ and from the right by $v$ yields $v^TP Q(x_0,\tau)v +v^TQ^T(x_0,\tau) P^T v = 0+0=0$, which contradicts \eqref{eq:lmi2}. Therefore, $Q(x_0,\tau)$ is invertible.
\end{proof}

For any fixed $(x_0,\tau)$, condition \eqref{eq:lmi} can be checked by solving an LMI feasibility problem. For practical numerical diagnosis over $\Phi_{\mathrm{back}}(\mathcal{C})\times[0,T]$, one may evaluate this condition on a sufficiently dense grid of $(x_0,\tau)$. Such sampling provides a numerical diagnostic rather than a formal certificate over the entire continuous domain.}

\begin{proposition}\label{proposition_equivalence}
	Consider the PWA system \eqref{closed_loop_pwa} and its two generalized sensitivities $Q_\mathrm{A}$ and $\partial_\mathrm{C} \phi_\mathrm{b}$. Suppose Assumption \ref{assumption_invert} holds. We have $\mathrm{conv}(Q_\mathrm{A} (x_0,\tau))= \partial_\mathrm{C} \phi_\mathrm{b}(x_0,\tau),\;\forall x_0 \in \mathbb{R}^n$ and $\forall \tau \in [0,\infty)$.
\end{proposition}

\begin{proof}
	See Appendix \ref{appendix_proposition_equivalence}.
\end{proof}

Proposition \ref{proposition_equivalence} indicates that the convex hull of the Aumann sensitivity, computed via \eqref{generalized sensitivity computation} can replace the Clarke sensitivity in the formulation of the proposed PSF \eqref{safety filter pwa}.

\textcolor{blue}{
\begin{remark}When Assumption \ref{assumption_invert} is not satisfied, the proof of Proposition \ref{proposition_equivalence} shows that the subset relation $\partial_C \phi_\mathrm{b}\left(x_0, \tau\right) \subseteq \operatorname{conv}\left(Q_\mathrm{A}\left(x_0, \tau\right)\right)$ still holds. What may be lost is the inverse relation. This means that replacing the Clarke sensitivity $\partial_C \phi_\mathrm{b}\left(x_0, \tau\right)$ in the proposed all-elements PSF in \eqref{safety filter pwa} with $\operatorname{conv}\left(Q_\mathrm{A}\left(x_0, \tau\right)\right)$ may introduce redundant constraints. Consequently, by further assuming that the all-elements PSF is recursively feasible, safety/forward invariance (Theorem \ref{theoreom_forward}) is preserved. However, since redundant constraints may be added, the feasibility of the all-elements PSF may be lost.
\end{remark}}

%\subsection{Computing the sensitivity}
%
%Based on the above analysis, the uniqueness of the Aumann sensitivity $Q_\mathrm{A}$ can be determined by computing either the backward or forward reachable set. In the offline setting, the backward reachable set can be used to identify the region $\Phi_{\mathrm{back}}(\mathcal{C})$ if $Q_\mathrm{A}$ is not unique. In the online setting, since we have the current state $\bar x_0$ at hand, computing the forward reachable set $\Phi_\mathrm{for}(\bar x_0)$, and checking whether it intersects with $\mathcal{C}$ are more convenient. Analytical computation of reachable sets for PWA systems is only possible in low-dimensional cases, such as the PWA system in Example 1. For high-dimensional PWA systems, the reachable sets can be approximated using the Hamilton–Jacobi framework \cite{bansal2017hamilton}.
%
%If $\Phi_\mathrm{for}(\bar x_0) \cap \mathcal{C} = \emptyset$, we can follow the standard procedure outlined in \cite[Chapter3.3]{khalil2002nonlinear} and use numerical integration to compute the unique sensitivity. Otherwise, we must calculate the set-valued sensitivity function based on \eqref{integral}. 

\section{Implementation of the all-elements PSF}\label{section implementation}

Problem \eqref{safety filter pwa} is a semi-infinite programming problem \cite[Chapter 5.4]{bonnans2013perturbation}, i.e., it contains infinitely many constraints. To solve it, we adopt the following three steps.

\textbf{Step 1: Transforming the constraints.} Under the equivalence between $\partial_\mathrm{C} \phi_\mathrm{b}$ and $\mathrm{conv}(Q_\mathrm{A})$, we replace $\partial_\mathrm{C} \phi_\mathrm{b}$ with $\mathrm{conv}(Q_\mathrm{A})$ in the formulation \eqref{safety filter pwa} of the proposed all-elements PSF. In other words, in \eqref{b} and \eqref{c}, the CBF constraints should be satisfied for any $Q \in \mathrm{conv}(Q_\mathrm{A})$. 

We first eliminate the infinite number of constraints stemming from the set-valued generalized Clarke derivatives. Since (i) the CBF constraints are linear in $h_X$, $h_\mathrm{b}$, and $Q$, and (ii) the sets $\partial_\mathrm{C} h_X$, $\partial_\mathrm{C} h_\mathrm{b}$, and $\mathrm{conv}(Q_\mathrm{A})$ are convex, according to \cite[Proposition 2.1]{ben1998robust}, \eqref{b} and \eqref{c} are equivalent to
\begin{align}\label{derobustify}
	&\partial h_X Q   f_\mathrm{PWA}(x,u)  \geq-\alpha\left(h_X\left(\phi_{\mathrm{b}}( x,\tau)\right)\right),\nonumber\\
	&\quad \quad \forall \partial h_X \in J_{h_X}(\phi_{\mathrm{b}}( x,\tau)),\;\forall  Q \in Q_\mathrm{A}(x,\tau),\; \forall \tau \in[0, T] ,\nonumber\\
	&\partial h_\mathrm{b} Q  f_\mathrm{PWA}(x,u)   \geq-\alpha_{\mathrm{b}}\left(h_{\mathrm{b}}\left(\phi_{\mathrm{b}}( x,T)\right)\right),\nonumber\\
	& \quad\quad \forall \partial h_\mathrm{b} \in J_{h_\mathrm{b}}(\phi_{\mathrm{b}}( x,T)),\;\forall  Q \in Q_\mathrm{A}(x,T),
\end{align}
where $J_{h_X}(\phi_{\mathrm{b}}( x,\tau)) =$
\begin{equation*}
	\left\{\lim _{i \rightarrow \infty} \left. D h_X(y_i) \right| y_i \rightarrow \phi_{\mathrm{b}}( x,\tau), \; h_X \text { is differentiable at } y_i\right\}
\end{equation*}
and $J_{h_\mathrm{b}}$ is defined similarly to $J_{h_X}$. According to Definition \ref{Clarke definition}, the set-valued functions $J_{h_\mathrm{b}}$ and $J_{h_X}$ denote the collections of limiting derivatives that generate $\partial_\mathrm{C} h_\mathrm{b}$ and $\partial_\mathrm{C} h_X$, respectively, i.e., $\partial_\mathrm{C} h_\mathrm{b} (\phi_{\mathrm{b}}( x,T)) = \mathrm{conv}(J_{h_\mathrm{b}}(\phi_{\mathrm{b}}( x,T)))$ and $\partial_\mathrm{C} h_X (\phi_{\mathrm{b}}( x,\tau)) = \mathrm{conv}(J_{h_X}(\phi_{\mathrm{b}}( x,\tau)))$. 

Note that $J_{h_\mathrm{b}}$ and $J_{h_X}$ are in general computable, as the explicit expressions of $h_\mathrm{b}$ and $h_X$ are usually known. Besides, the Aumann sensitivity $Q_A$ is computed through \eqref{generalized sensitivity computation}. Note that $|J_{h_\mathrm{b}}|$, $|J_{h_X}|$, and $|Q_\mathrm{A}|$ are finite. \textcolor{blue}{For polytopic $X$, let $X:=\{x\in\mathbb{R}^n \mid C_X x\leq c_X\}$ denote its
hyperplane representation, where $C_X$ is a matrix and $c_X$ is a vector of appropriate dimensions. The function $h_X$ can be defined as $h_X:= \min_{i=1,2,...,\dim(c_X)}\{(c_X)_i-(C_X)_i x\}$ so that $h_X (x)\geq0 \Leftrightarrow x\in X$. Then, $J_{h_X}$ can be computed by $J_{h_X}(x)=\bigl\{-(C_X)_i\;\big|\; i \text{ s.t. } h_X(x)=(c_X)_i-(C_X)_ix\bigr\}$.}

\textbf{Step 2: Mixed-integer encoding of $f_\mathrm{PWA}$.} Mixed-integer encoding is commonly used in hybrid MPC in which the PWA prediction model is usually converted into an equivalent mixed-logical dynamical form \cite{borrelli2017predictive}. 

If the partition of the PWA system \eqref{pwa_form} is only in the state space, i.e., $f_{\mathrm{PWA}}(x, u):=A_i x+B_i u+c_i, \text { for } x \in \mathcal{P}_i$, the inequalities in \eqref{derobustify} are naturally linear in $u$ and hence there is no need to do mixed-integer encoding.

If the partition is in both state and input spaces, to make the inequalities in \eqref{derobustify} linear in $u$, we follow \cite{bemporad1999control} to use the Big-M approach to get the equivalent mixed-integer formulation of $f_{\mathrm{PWA}}(x, u)$, which is given by 
\begin{align}\label{eq:bigm}
	&f_{\mathrm{PWA}}(x, u) = \sum\limits_{i = 1}^{p(x)} {{f_i}} ,\;\sum\limits_{i = 1}^{p(x)} {{\Delta _i}}  = 1 \nonumber \\
	&{\text{for}}\;i = 1,2,...,p(x): 
	\left\{ \begin{gathered}
		{f_i} \leq \overline{M}{\Delta _i} \hfill \\ 
		{f_i} \geq \underline{M}{\Delta _i} \hfill \\
		{f_i} \leq {A_i}x + {B_i}u + {c_i} - \underline{M}\left( {1 - {\Delta _i}} \right) \hfill \\
		{f_i} \geq {A_i}x + {B_i}u + {c_i} - \overline{M}\left( {1 - {\Delta _i}} \right) \hfill \\
		{\Psi^{(x)}_i}x + {\Psi^{(u)}_i}u \leq {\psi_i} + \overline{M}_i^*(1 - {\Delta _i}) 
	\end{gathered}  \right.
\end{align}
\textcolor{blue}{Let $\mathcal{F}(x)$ denote the fiber partition of the partition $\mathcal{P}$ of the PWA system over the input space for a given $x$. In particular, $\mathcal{F}(x):=\left\{\mathcal{F}_i(x)\;|\; i \in \mathbb{N}^+_p, \mathcal{F}_i(x) \neq \varnothing\right\}$, where $\mathcal{F}_i(x) := \{ u \in \mathbb{R}^{m} | {\Psi^{(x)}_i}x + {\Psi^{(u)}_i}u \leq {\psi_i} \}=\left\{u \in \mathbb{R}^{m}|\left(x, u\right) \in \mathcal{P}_i\right\}$. In \eqref{eq:bigm}, $p(x)$ is the cardinality of $\mathcal{F}$, and $\Delta _i \in \{0,1\}$ are binary variables.} \textcolor{blue}{Three bounds $\underline{M}$, $\overline{M}$, and $\overline{M}_i^*$ in \eqref{eq:bigm} need to be determined. According to \cite{bemporad1999control}, the expressions of these bounds are computed by solving the linear programs: $\underline{M}_j = \min_{i\in \mathbb{N}_p^+} \min_{x\in X,u\in U} (A_i)_jx + (B_i)_ju + (c_i)_j$, $\overline{M}_j = \max_{i\in \mathbb{N}_p^+} \max_{x\in X,u\in U} (A_i)_jx + (B_i)_ju + (c_i)_j$, and $\overline{M}_i^* = \max_{x\in X,u\in U} {\Psi^{(x)}_i}x + {\Psi^{(u)}_i}u - {\psi_i}$. For any bounds that are no smaller than those computed by the above linear programs, the mixed-logical-dynamical equivalence \eqref{eq:bigm} is exact. Unnecessarily large bounds return the same set of optimal solutions but may slow the solver, while tighter bounds may speed it up. This is why the LP-based (tightest) bounds are preferred.}

\textbf{Step 3: Discretizing the time $\tau$.} After Steps 1 and 2, the infinite number of constraints in \eqref{derobustify} comes only from the continuous time set $[0,T]$. In practice, to obtain computational tractability, the constraints are usually relaxed into a finite number of constraints that should be satisfied at some time instants \cite{molnar2023safety}. In particular, given the prediction horizon $T>0$ and a number $N \in \mathbb{N}^+$, the time set $[0,T]$ is discretized into $\{ \tau_{(l)}\}^N_{l=0}$, where each $\tau_{(l)} = lT/N$. As a result, \eqref{derobustify} is relaxed into an inequality constraint that should be satisfied at finitely many $\tau_{(l)}$. \textcolor{blue}{To ensure constraint satisfaction between sampling instants, following \cite{van2024disturbance}, we tighten the constraint \eqref{b} by a constant $\varepsilon$, where $\varepsilon$ should satisfy $\varepsilon \geq \frac{T}{2N} L_{h_X} \sup_{x\in X} \|f_\mathrm{PWA,b}(x)\|_2$, with $L_{h_X}$ the Lipschitz constant of $h_X$ w.r.t. the Euclidean norm. This tightening bound, which depends on the grid size, the Lipschitz constant of the constraint function, and a bound on the backup closed-loop vector field, guarantees that the controller after time discretization remains safe.}

\begin{remark}[Computational complexity]
\textcolor{blue}{Let $\mathcal{F}(x)$ denote the fiber partition of the partition $\mathcal{P}$ of the PWA system over the input space for a given $x$. In particular, $\mathcal{F}(x):=\left\{\mathcal{F}_i(x): i \in \mathbb{N}^+_p, \mathcal{F}_i(x) \neq \varnothing\right\}$, where $\mathcal{F}_i(x) := \left\{u \in U:\left(x, u\right) \in \mathcal{P}_i\right\}$. The number of binary variables for problem \eqref{safety filter pwa} is equal to the cardinality $p(x)$ of $\mathcal{F}(x)$, which is at most $p$ and independent of the time horizon $T(N)$. In comparison, hybrid MPC that uses the big-M mixed-integer encoding introduces $pN$ binary variables \cite{borrelli2017predictive}. Since MIQP solve time is governed in the worst case by the number of binary variables \cite{richards2005mixed}, solving \eqref{safety filter pwa} is in general more computationally efficient than solving a hybrid MPC problem, particularly for long prediction horizons.}

\textcolor{blue}{In the worst case, the cardinality of $Q_\mathrm{A}(x,\tau)$ is $\prod_{k:\,\tau_k\le\tau}|I_k|$. By Proposition \ref{proposition1} $|I_k|>1$ only on a critical period (a boundary-riding interval with $\phi_\mathrm{b}(x_0,\cdot)\in \mathcal{C}$). Hence, with
$M_\mathrm{cr}$ the number of critical periods in $[0,T]$ and $\beta:=\max_k|I_k|$, one has $|Q_\mathrm{A}(x_0,\tau)|\le\beta^{M_\mathrm{cr}}$. The discretized PSF then contains $O(NN_h\beta^{M_\mathrm{cr}})$ scalar inequalities, where $N_h$ is the maximum of $|J_{h_X}|$ along the prediction horizon. In the worst case $\beta = r$ and $M_\mathrm{cr}=M$, giving $|Q_\mathrm{A}|\le r^{M}$. However, note that this only affects the number of constraints in the PSF problem, while the number of integer variables does not depend on $|Q_\mathrm{A}|$. Methods to prevent $|Q_\mathrm{A}|$ from growing excessively include redesigning the backup controller to reduce boundary-riding behavior if excessive branching occurs, or shortening the prediction horizon $T$ so that $M_\mathrm{cr}$ is reduced.}
\end{remark}

\begin{remark}[Comparison with MPC on conservatism]
    \textcolor{blue}{The proposed PSF seeks the input closest to an arbitrary reference policy while enforcing safety constraints, whereas hybrid MPC optimizes a cumulative horizon cost subject to model and safety constraints. Therefore, a direct scalar comparison of “conservatism” is objective-dependent. Nevertheless, the proposed filter may be viewed as less intrusive with respect to the prescribed reference policy, while hybrid MPC has greater freedom to trade short-term deviations against future performance.}
\end{remark}

%\begin{remark}
%	Note that, in contrast to Steps 1 and 2, which involve equivalent transformations, Step 3 introduces a relaxation of the constraints in \eqref{b}. To ensure the safety of the resulting filtered controller following this relaxation, it is typically necessary to tighten the inequalities in \eqref{b} and \eqref{c}. We refer the reader to \cite{van2024disturbance} for a detailed discussion on this topic.
%\end{remark}

{\color{blue} \textbf{Lipschitz continuity of the solution.} If problem \eqref{safety filter pwa} leads to a convex QP (which is common in practical examples such as adaptive cruise control \cite{corona2008adaptive} and contact dynamics in robot manipulation \cite{marcucci2017approximate}), based on \cite{MESTRES2025101098}, we can derive a sufficient condition for local Lipschitzness of the solution, which is required in Theorem \ref{theoreom_forward}. Since the objective is strongly convex, the minimizer $u^*$ is unique for every $x$. Let $\bar B(x)$ denote the matrix of coefficients of $u$ in the inequality constraints of the PSF, after the finite-constraint reformulation and the time discretization. The sufficient condition is the Linear Independence ConstraintvQualification (LICQ): for all $x\in\Phi_{\mathrm{const}}(S_\mathrm{b},T)$, the rows
$\{\bar B_i(x)\;|\; i\in\mathcal{A}(x,u^*(x))\}$ should be linearly independent, where
$\mathcal{A}(x,u^*(x))$ is the index set of the constraints active at $x$ and the minimizer $u^*(x)$.

When problem \eqref{safety filter pwa} leads to an MIQP, the uniqueness and continuity of the solution are not guaranteed. The closed-loop dynamics
can instead be interpreted in the sense of Filippov. Let
$F\;:\;x\to F(x)$ denote the Filippov set-valued map of the closed-loop field (see \cite{filippov1964differential} for the detailed definition). A sufficient condition ensuring
that every Filippov solution starting in $\Phi_{\mathrm{const}}(S_\mathrm{b},T)$ remains in it is
\begin{align}\label{eq:Filippov2}
     &\partial h_X Q f  \geq-\alpha\left(h_X\left(\phi_{\mathrm{b}}( x,\tau)\right)\right),\;\forall f \in F(x),  \nonumber\\
		&\quad \;\;\forall \partial h_X \!\in \!\partial_\mathrm{C} h_X(\phi_{\mathrm{b}}( x,\tau)),\;\forall  Q \!\in\! \partial_\mathrm{C}\phi_\mathrm{b}(x,\tau),\; \forall \tau \!\in\! [0, T]  \nonumber \\
		&\partial h_{\mathrm{b}} Q  f   \geq-\alpha_{\mathrm{b}}\left(h_{\mathrm{b}}\left(\phi_{\mathrm{b}}( x,T)\right)\right),\nonumber\\
		&\quad \;\;\forall f \in F(x),\;\forall \partial h_{\mathrm{b}} \in \partial_\mathrm{C}  h_{\mathrm{b}}(\phi_{\mathrm{b}}( x,T)),\; \forall  Q \in \partial_\mathrm{C}\phi_\mathrm{b}(x,T).
\end{align}
Compared with the original CBF constraints in \eqref{safety filter pwa}, \eqref{eq:Filippov2} further requires the CBF inequalities to hold for every element $f$ of the Filippov set $F(x)$. Forward invariance then follows from a differential-inclusion counterpart of Lemma \ref{lemma_invariance}. In particular, since $f$ is upper bounded in $X \times U$, the Filippov map $F$ is upper semi-continuous, nonempty, and convex \cite{filippov1964differential}. For any Filippov solution $x$, we have $\dot x(t) \in F(x(t)),\text{ a.e. } t$. Since the CBF inequalities hold for every $f \in F(x)$, the differential inequality used in the proof of Lemma \ref{lemma_invariance}. holds along every Filippov solution a.e. Therefore, the comparison theorem argument applies element-wise to $F(x)$.}

\section{Explicit approximation of the safe controller}\label{section_approximate}
In the previous section, we have shown that the intractable constraints in \eqref{safety filter pwa} can be reformulated using Aumann sensitivity. However, When the partition of the PWA system \eqref{pwa_form} is defined over the joint state-input space, the computation time for solving the MIQP derived from \eqref{safety filter pwa} in practice grows dramatically with the number of binary variables in the worst case \cite{richards2005mixed}. The number of binary variables itself scales linearly with the number $p$ of regions. Consequently, the proposed all-elements PSF becomes impractical for real-time control systems with limited computational resources. To address this challenge (Problem 3), we further derive an explicit approximation for the solution to \eqref{safety filter pwa}. This explicit solution is guaranteed to be feasible w.r.t. \eqref{safety filter pwa} if the original problem is feasible, and can be computed efficiently.

Our design is inspired by \cite{tordesillas2023rayen}, where a constrained neural network is employed to approximate solutions of convex constrained optimization problems. An important property of \eqref{safety filter pwa} we use is that the backup solution $\pi_\mathrm{b}(x)$ is a feasible solution to \eqref{safety filter pwa} for any $x \in \Phi_{\text {const}}\left(S_{\mathrm{b}}, T\right)$, as has been proved in Theorem \ref{Theorem_feasibility}. 

For each $x$, define by $U_\eqref{safety filter pwa}(x)$ the feasible set of \eqref{safety filter pwa}. First, we  restrict the constraint on $u$ from $u$ belonging to the whole feasible set $U_\eqref{safety filter pwa}(x)$ to 
\begin{equation}\label{tighten}
	u\in U_\eqref{safety filter pwa}(x) \cap U_\mathrm{line}(x),
\end{equation}
where $U_\mathrm{line}(x)$ is the line segment with the endpoints $\pi_\mathrm{b}(x)$ and $\pi_\mathrm{r}(x)$. Mathematically, $U_\mathrm{line}(x) = \{ u\in \mathbb{R}^m | u = \pi_\mathrm{b}(x) + \lambda(\pi_\mathrm{r}(x) -\pi_\mathrm{b}(x)),\;\lambda \in [0,1] \}$.

By adding the constraint $u \in U_\mathrm{line}(x)$ into the optimization problem \eqref{safety filter pwa}, we obtain that the optimizer $\hat\pi_\mathrm{safe}(\cdot): \mathbb{R}^n \to \mathbb{R}^m$ of the modified problem takes the following form:
\begin{equation}\label{approximated_policy}
	\hat\pi_\mathrm{safe}(x) = \pi_\mathrm{b}(x) + \lambda^*(x)(\pi_\mathrm{r}(x) -\pi_\mathrm{b}(x)),
\end{equation}
where 
\begin{align}\label{lamda}
	\lambda^*(x) =& \max_{\lambda \in [0,1]} \;\lambda\nonumber\\
	& \text{ s.t. } u = \pi_\mathrm{b}(x) + \lambda(\pi_\mathrm{r}(x) -\pi_\mathrm{b}(x)),\; u \in U_\eqref{safety filter pwa}(x).
\end{align}

Since each polyhedral region $\mathcal{P}_i$ has the hyperplane representation $\{ [x^T,\;u^T]^T \in \mathbb{R}^{n+m} | {\Psi^{(x)}_i}x + {\Psi^{(u)}_i}u \leq {\psi_i} \}$, for any state $x$, we can determine the slice $U_i(x)$ of $\mathcal{P}_i$, defined by
\begin{equation}\label{ui}
	U_i(x) : = \{ u \in \mathbb{R}^m | {\Psi^{(u)}_i}u \leq  \psi_i - \Psi^{(x)}_i x\}.
\end{equation}

Let $I_u(x)$ denote the set of indices $i$ for which $U_i(x)$ is not empty. By substituting the expression of $f_\mathrm{PWA}$ into \eqref{safety filter pwa} and applying Step 1 of Section \ref{section implementation}, we can rewrite the optimizer $\lambda^*$ in a more compact form:
\begin{align}\label{lamda compact}
	\lambda^*(x) =& \max_{\lambda } \;\lambda\nonumber\\
	& \text{ s.t. } \exists i \in I_u(x) \text{ s.t. } \eta_i(x,\tau) \lambda \leq \omega_i(x,\tau),\;\forall \tau \in [0,T],
\end{align}
where $\eta_i$ and $\omega_i$ are vector-valued functions with proper dimensions. The expressions of $\eta_i$ and $\omega_i$ can be obtained based on \eqref{pwa_form}, \eqref{derobustify}, and \eqref{lamda}. To derive the explicit form of $\lambda^*$, the constraint $\lambda \in [0,1]$ is incorporated in \eqref{lamda compact} with the form of $[1\;\;-1]^T \lambda \leq [1\;\;0]^T$. 

The following proposition gives the explicit form of the optimizer $\lambda^*$.

\begin{proposition}\label{proposition_explicit}
	Consider the optimization problem \eqref{lamda}. \textcolor{blue}{Suppose that the conditions of Theorem \ref{Theorem_feasibility} hold. Then, \eqref{lamda} is feasible for any $x \in \Phi_{\text {const}}\left(S_{\mathrm{b}}, T\right)$ and any $T\geq0$}, and the optimizer $\lambda^*$ has the expression:
	\begin{equation}\label{lamda expression}
		\lambda^*(x) = \max_{i\in  I_u(x)} \min_{\tau \in[0,T]}\min_{j \in \mathbb{N}^+_{\mathrm{dim}(\eta_i)}} \;\lambda_{i,j} (x,\tau),
	\end{equation}
	with
	\begin{equation}\label{lamda ij}
		\begin{aligned}
		& \lambda_{i,j} (x,\tau) \\
		&= 
		\left\{\begin{aligned}
			&\frac{\omega_{i,j}}{\eta_{i,j}},\;  \text { if }  \eta_{i,j}>0\\
			&\mathrm{step} (\omega_{i,j}) ,\;  \text { if } \eta_{i,j}= 0\\
			&\mathrm{step} \left(\omega_{i,j} \!-\!\eta_{i,j} \min_{ \lambda_{i,j'}(x,\tau')\! \in\! (0,1)}  \lambda_{i,j'} (x,\tau')  \right),\;  \text { if } \eta_{i,j}\!<\!0,
		\end{aligned}\right.
	\end{aligned}
	\end{equation}
	where $\omega_{i,j}$ and $\eta_{i,j}$ represent the $j$-th element of $\omega_i$ and $\eta_i$, respectively, and we omit their dependence on $x$ and $\tau$ for notational brevity.
\end{proposition}

\begin{proof}
	See Appendix \ref{appendix_proposition_explicit}.
\end{proof}

Since $\hat{\pi}_\mathrm{safe}$ is a feasible solution to the proposed all-elements PSF \eqref{safety filter pwa}, we immediately have the following corollary of Theorem \ref{theoreom_forward} and Proposition \ref{proposition_explicit} on the safety performance of $\hat{\pi}_\mathrm{safe}$:

\begin{corollary}
	Consider the continuous PWA system \eqref{pwa_form} and the proposed all-elements PSF \eqref{safety filter pwa}. Suppose that Assumptions \ref{A1} and \ref{assumption_invert} hold and that the approximate policy $\hat{\pi}_\mathrm{safe}$ in \eqref{approximated_policy} with $\lambda^*$ computed from \eqref{lamda expression} is locally Lipschitz continuous. Then, the approximate policy $\hat{\pi}_\mathrm{safe}$ renders $\Phi_{\text{const}}(S_\mathrm{b},T)$ forward invariant. 
\end{corollary}

%Similar to Lemma \ref{lemma_safety}, to show the safety of $\hat{\pi}_\mathrm{safe}$, the corollary assumes the Lipschitz continuity of $\hat{\pi}_\mathrm{safe}$. Analyzing the conditions under which $\hat{\pi}_\mathrm{safe}$ is Lipschitz continuous is left for future work.

Akin to Step 3 of Section \ref{section implementation}, implementing the approximate policy \eqref{approximated_policy} needs the discretization of the time set $[0,T]$ appearing in \eqref{lamda expression}.

{\color{blue} \textbf{Sub-optimality.} If the proposed PSF is cast as an MIQP, it is difficult to quantify the suboptimality of the approximate solution because of its nonconvex nature. Therefore, we focus on the case when \eqref{safety filter pwa} reduces to a convex QP. The exact optimizer is the Euclidean projection $\pi_{\mathrm{safe}}(x)=\Pi_{U_{(11)}(x)}(\pi_r(x))$, while $\hat\pi_{\mathrm{safe}}(x)$ is the point of $U_{\mathrm{line}}(x)$ nearest $\pi_r(x)$. The projection inequality then yields the following corollary:
\begin{corollary}\label{corollary:suboptimal}
    Consider the explicit policy \eqref{approximated_policy}, if the PWA system \eqref{pwa_form} is partitioned only w.r.t. the state space, then the suboptimality of \eqref{approximated_policy} is upper bounded by
\begin{align*}
   &\max\{ \bigl\|\hat\pi_{\mathrm{safe}}(x)-\pi_{\mathrm{r}}(x)\bigr\|,\;\bigl\|\hat\pi_{\mathrm{safe}}(x)-\pi_{\mathrm{safe}}(x)\bigr\|\}\\
   &\leq \bigl(1-\lambda^*(x)\bigr)\,\bigl\|\pi_\mathrm{r}(x)-\pi_\mathrm{b}(x)\bigr\|. 
\end{align*}
\end{corollary}
The bound in Corollary \ref{corollary:suboptimal} is informative in two ways: (i) it vanishes as $\lambda^*(x)\to 1$, i.e., whenever the reference is close to feasible; and (ii) the tracking error relative to the reference is never larger than that of the backup policy.}

\begin{remark}
	Unlike \cite{tordesillas2023rayen}, which computes explicit approximations for convex optimization problems, our approach accommodates non-convex MIQPs. Compared to our previous work \cite{he2024efficient}, which also finds explicit approximations for MIQP problems, we extend the framework to problems with infinitely many constraints.
\end{remark}

\section{Case studies}
\textcolor{blue}{We consider two examples to demonstrate the effectiveness of the proposed PSF policy. In the first example, we consider an inverted pendulum, in which $Q_\mathrm{A}$ and CBF gradients are set-valued. We demonstrate that the proposed all-elements PSF achieves enhanced constraint satisfaction compared with standard PSFs that employ only a single gradient. The second example considers a four-room temperature-control system whose PWA partition is defined over both the state and input spaces, such that the resulting PSF is formulated as a computationally expensive MIQP. This example is used to demonstrate the computational improvement of the proposed explicit approximation.}
\subsection{Inverted pendulum}
\begin{figure}
	\centering
	\includegraphics[width=80pt,clip]{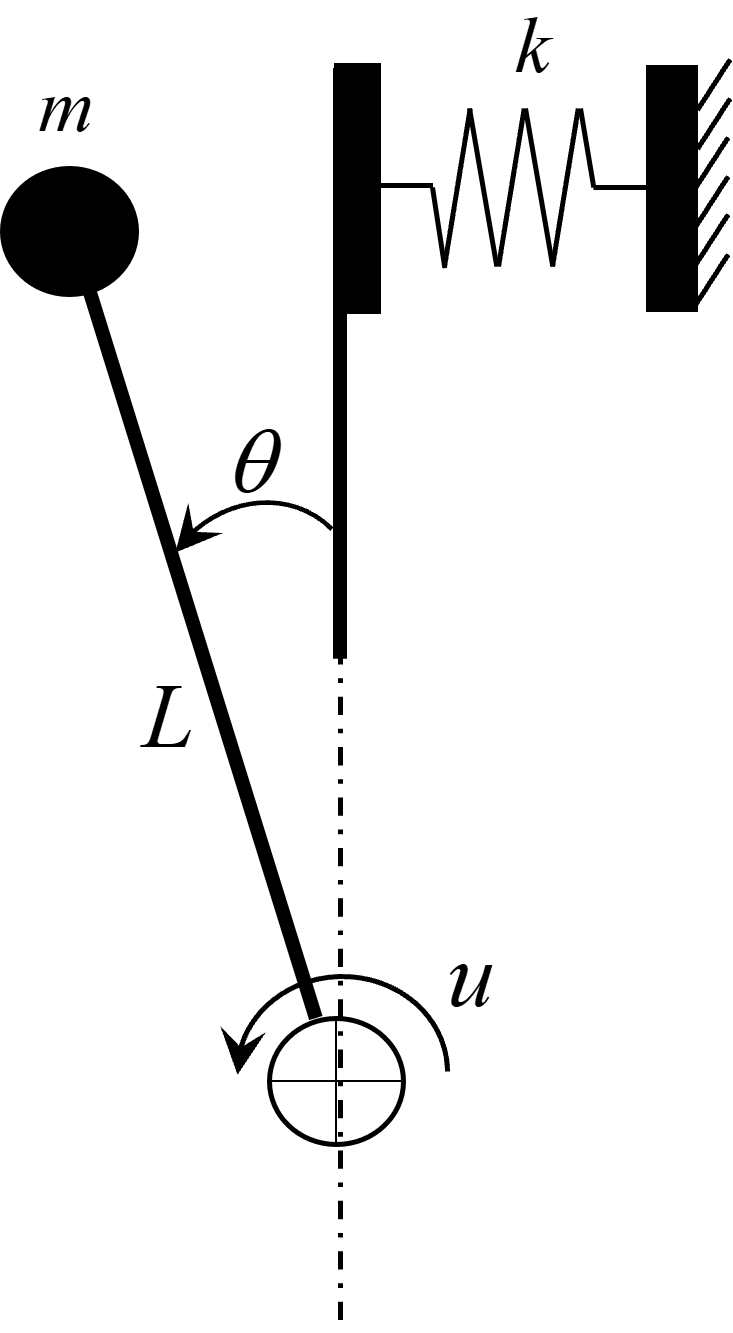}
	\caption{Diagram of the inverted pendulum interacting with an elastic wall.}
	\label{fig_pendulum}
\end{figure}
To demonstrate the theoretical results, control of a simple inverted pendulum that interacts with an elastic wall (illustrated in Fig. \ref{fig_pendulum}) is considered. The system is linearized around the upright position, resulting in the following PWA system:

\begin{equation*}
	\ddot{\theta}  = \begin{cases}
		\frac{g}{L} \theta + \frac{1}{mL^2} u,\;&\text{if } \theta\geq 0\\
		(\frac{g}{L} - k L) \theta + \frac{1}{mL^2} u,\;&\text{if } \theta< 0,
	\end{cases}
\end{equation*}
where $\theta$ is the pendulum angle, $u$ is the input torque, the state variable $x=[\theta\;\;\dot\theta]^T$, $m=1$ kg, $L=1$ m, $g=10$ m/s$^2$, and $k = 2$ N/m is the stiffness of the wall. The state constraint is $x \in X:=\{x|\;|\theta| \leq 0.5,\; |\dot{\theta}| \leq2 \}$, and the input constraint is $|u|\leq10$. In the simulation, the system is discretized by forward Euler with a sampling time of 0.001 seconds.

The backup policy is designed as a linear feedback controller, given by $\pi_\mathrm{b}(x) = -12 \theta - 3 \dot{\theta}$. Under this policy, the resulting PWA system takes the form:
\begin{equation}\label{eq_inverted_closed}
	\dot x = f_\mathrm{PWA,b}(x) = \begin{cases} \left[\begin{array}{rr}
			0 & 1 \\
			-2 & -3
		\end{array} \right]x ,\;&\text{if } x_1\geq 0\\
		\left[\begin{array}{rr}
			0 & 1 \\
			-4 & -3
		\end{array}\right] x ,\;&\text{if } x_1< 0.
	\end{cases}
\end{equation}
The feedback gain $[-12\;-3]$ is selected such that both linear modes of \eqref{eq_inverted_closed} are Hurwitz stable.

According to Proposition \ref{proposition1}, $x = [0\;\;0]^T$ is the only initial point at which the Aumann sensitivity contains multiple (two in this case) elements. Under the backup policy, a quadratic CBF $h_\mathrm{b}(x) = \gamma - x^T R x$ is synthesized using the LMI approach presented in \cite{wabersich2022predictive}. The reference control policy is chosen as $\pi_{\mathrm{r}}(x) = -12 (\theta -\theta_\mathrm{r})  - 3 (\dot{\theta} - \dot{\theta_\mathrm{r}}) + mL^2 \ddot{\theta_\mathrm{r}} - mgL \theta_\mathrm{r}$, which is a proportional-derivative plus feedforward controller. The reference angle is chosen as $\theta_\mathrm{r}(t) = 0.8\cos(t)$. This means that during some periods, the reference state violates the state constraint. For the PSF in \eqref{safety filter pwa}, the horizon $T$ is set to 1 second and $[0,\;T]$ is further discretized into $N = 50$ intervals, resulting in the time grid $\{lT/N\}^N_{l=0}$.

\textcolor{blue}{In Fig. \ref{fig:invertible}, we verify the invertibility of the sensitivity matrix over the horizon $[0,T]$ by displaying its determinant. One can observe that $\mathrm{det}(\mu Q_1+(1-\mu) Q_2)>0$ for all $\mu\in[0,1]$. This provides numerical evidence supporting Assumption \ref{assumption_invert} for the inverted-pendulum example.}
\begin{figure}
    \centering
    \includegraphics[width=0.4\textwidth]{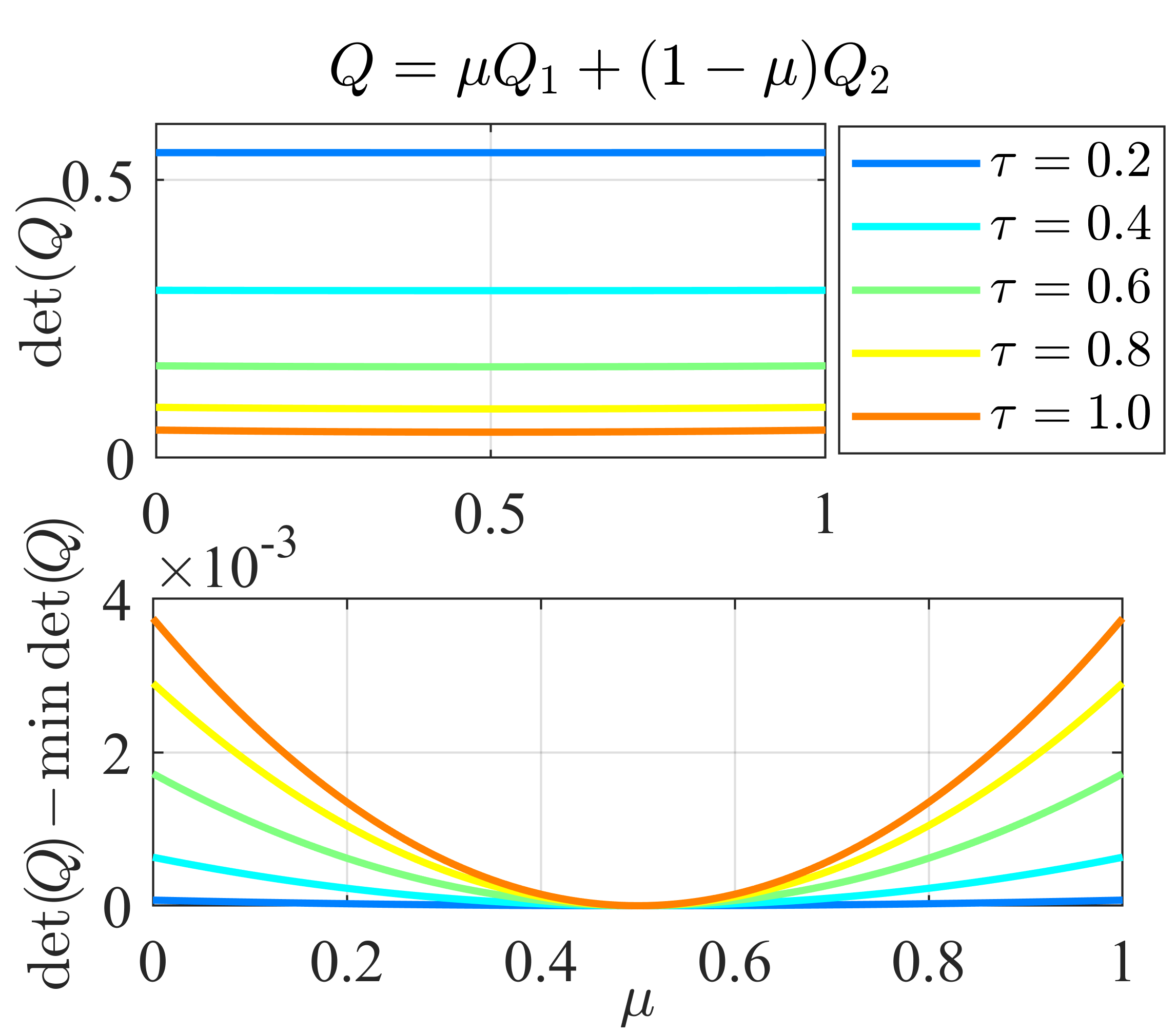}
    \caption{\textcolor{blue}{The determinant of $Q \in \mathrm{conv}(Q_\mathrm{A}(0,\tau))$ for the sampled $\tau$ from $[0,1]$ is always positive.}}
    \label{fig:invertible}
\end{figure}

\textcolor{blue}{The inverted-pendulum example compares three controllers:
\begin{itemize}
    \item The standard CBF safety filter \cite{ames2016control}, which is equivalent to the predictive one with the prediction horizon $T=0$.
    \item The classical (single-gradient) PSF \cite{molnar2023safety}, which includes only a single, randomly selected gradient from $\partial h_X$.
    \item The proposed all-elements PSF.
\end{itemize}}

Fig. \ref{time_traj} compares the time responses of the closed-loop system under the PSF \eqref{safety filter pwa} and the standard CBF safety filter. The initial state is $[0.1\;\;0.1]^T$. From Fig. \ref{time_traj}, it is observed that even though the reference state violates the state constraint, the trajectories controlled by these two safety filters always satisfy the state constraint. The PSF is less conservative because the corresponding trajectory stays closer to the reference. Furthermore, Fig. \ref{2D_traj} displays the trajectories and also the safe set of $h_\mathrm{b}$. It is found that the feasible region of the PSF extends beyond the safe set of $h_\mathrm{b}$ (i.e., the feasible region of the standard safety filter), highlighting the superiority of the predictive approach over the standard method.

\begin{figure}
	\centering
	\includegraphics[width=240pt,clip]{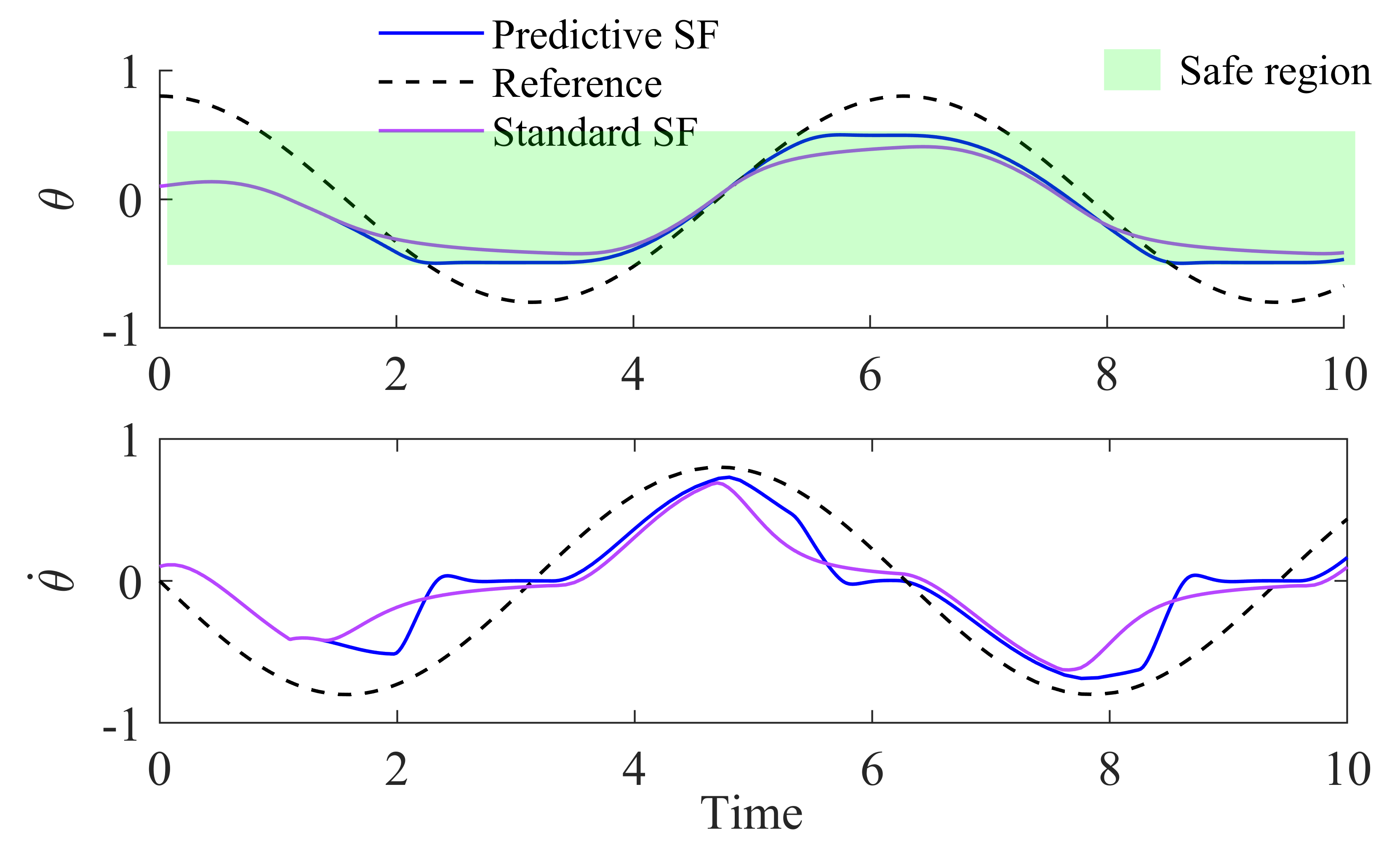}
	\caption{Time responses of the closed-loop system with the PSF \eqref{safety filter pwa} and the standard safety filter. The initial state is $[0.1\;\;0.1]^T$. “SF” means “safety filter”. Both SFs successfully control the system without recording any constraint violations. \textcolor{blue}{The worst-case constraint value, $\min_{t\in[0,10]} h_X(x(t))$, is 0.0766 for the standard SF and 0.0007 for the PSF.}}
	\label{time_traj}
\end{figure}

\begin{figure}
	\centering
	\includegraphics[width=240pt,clip]{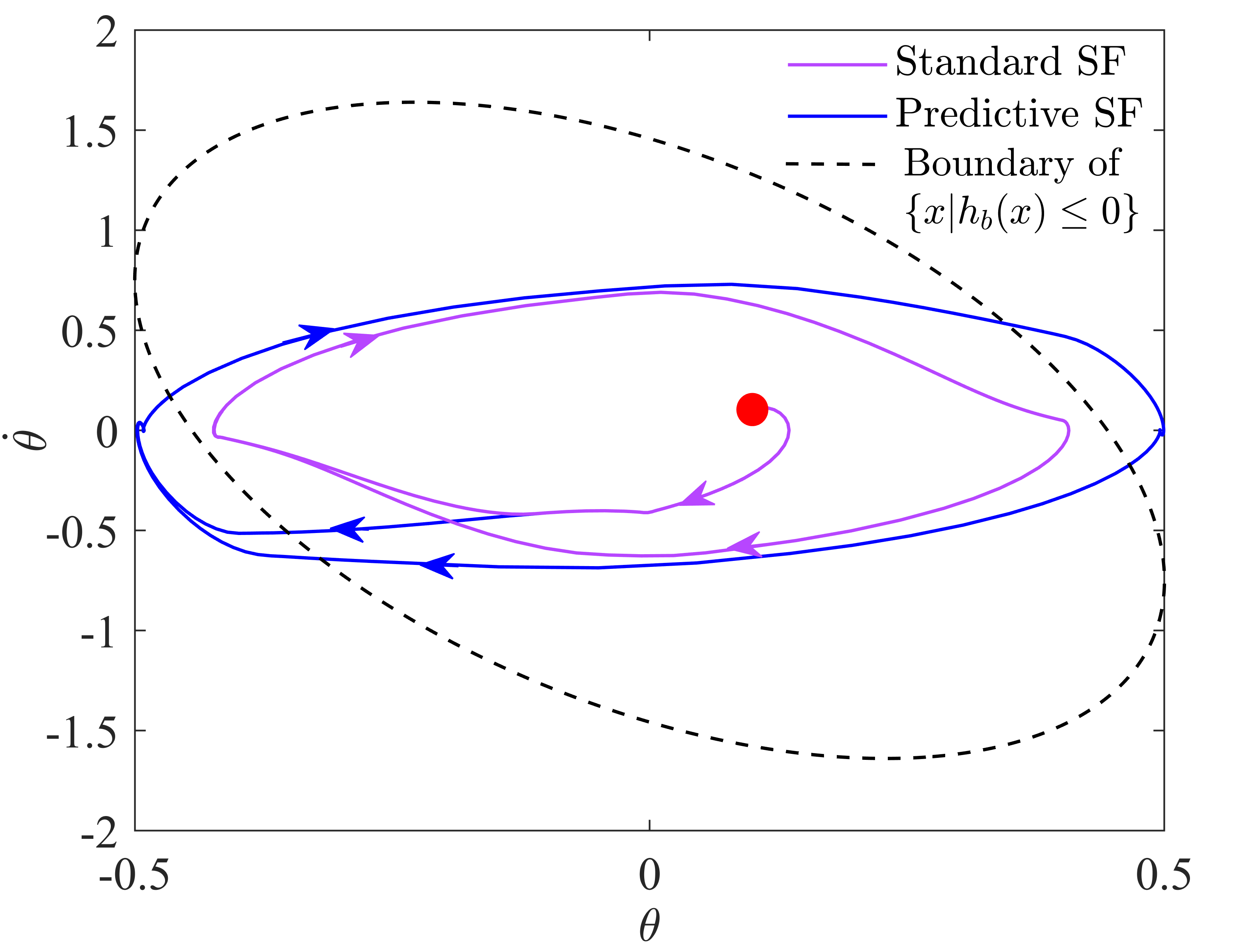}
	\caption{Trajectories of the closed-loop system (solid lines) and the safe set of $h_\mathrm{b}$ (black dashed line). The black dot represents the initial state. By virtue of predictions, the predictive SF enlarges the original feasible region $\{x | h_b(x) \geq 0\}$.} 
	\label{2D_traj}
\end{figure}

Next, we examine a scenario in which the PSF based on the classical gradient \eqref{sensitivity} fails to ensure safety. Consider the initial state $x(0) = [0.5\;\;-2]^T$, which lies on the boundary of $X$. The function $h_X$, defined as $h_X(x) :=  \min\{0.5-\theta,\; \theta+0.5,\;2 - \dot{\theta} ,\; \dot{\theta}+2 \}$, is not differentiable at $x(0)$. We design two PSFs: the first is based on \eqref{safety filter pwa}, and the second is a modified version of \eqref{safety filter pwa} that includes only a single, randomly selected gradient from $\partial_\mathrm{C}h_X$. The second one corresponds to the classical PSF, which implicitly assumes that $h_X$ is differentiable everywhere. To highlight the difference between the two safety filters, a constant reference policy is chosen as $\pi_{\mathrm{r}}(x) = -10$. 

Fig. \ref{Clarke_gradient_traj} illustrates the closed-loop trajectory of $\dot{\theta}$ under both safety filters. It is evident in the zoomed figure that the proposed all-elements PSF using the generalized Clarke derivative (blue curve) keeps the system within the safe region, while the safety filter using the single gradient (purple curve), proposed in \cite{gurriet2020scalable,chen2021backup}, experiences constraint violations \textcolor{blue}{between $t=0$ and $t=0.05$. When the safety filter with the single gradient becomes infeasible, we switch to applying the reference $\pi_{\mathrm{r}}$ through a saturation unit. However, this fallback does not prevent violation of the state constraint. A safety-preserving fallback is possible only when the current state lies in a certified, larger invariant set of a backup controller. In that case, the backup controller can be applied to preserve safety.}

\begin{figure}
	\centering
	\includegraphics[width=240pt,clip]{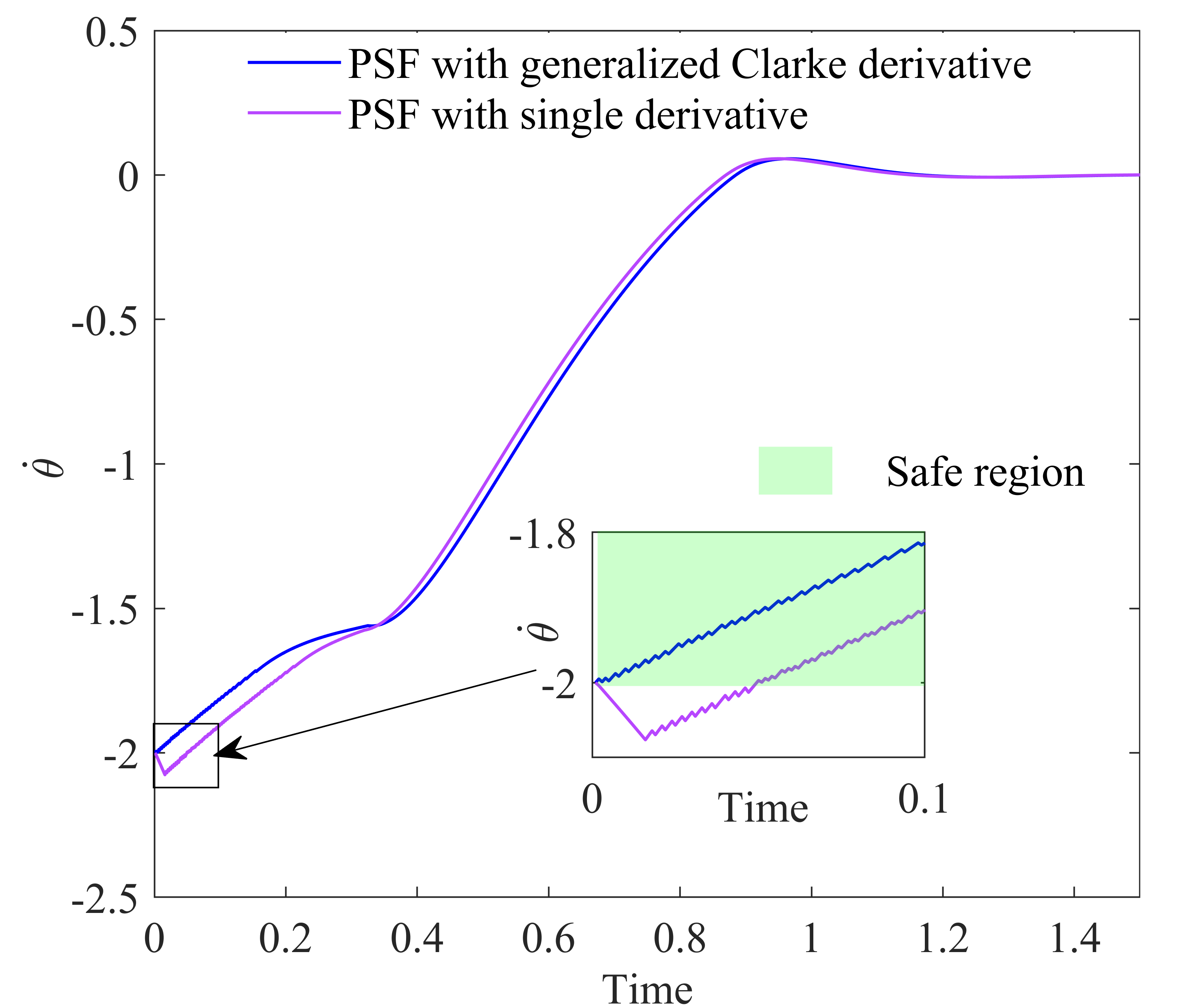}
	\caption{Comparison of PSFs using generalized Clarke derivative and single derivative methods. The highlighted inset zooms into the initial phase of the trajectory, illustrating the difference in early behavior between the two methods. When the initial state is at the boundary of the state constraint set, the proposed all-elements PSF using generalized Clarke derivative successfully avoids the unsafe region. \textcolor{blue}{The classical PSF proposed in \cite{gurriet2020scalable,chen2021backup} fails to guarantee safety. Specifically, the constraint $h_X(x(t)) \geq 0$ is violated during the interval $[0,0.05]$, with the worst-case violation given by $\min_{t\in[0,10]} h_X(x(t)) = -0.076$.}}
	\label{Clarke_gradient_traj}
\end{figure}

\subsection{Temperature control}
To demonstrate the computational advantage of the proposed explicit solution, we further consider a more complex PWA system with more partitions. In particular, let us consider control of the temperatures $x = [T_1\; T_2\; T_3\;T_4]^T$ of four rooms in a building. The layout of the four rooms and heat flows are shown in Fig. \ref{fig_room}.

\begin{figure}
	\centering
	\includegraphics[width=240pt,clip]{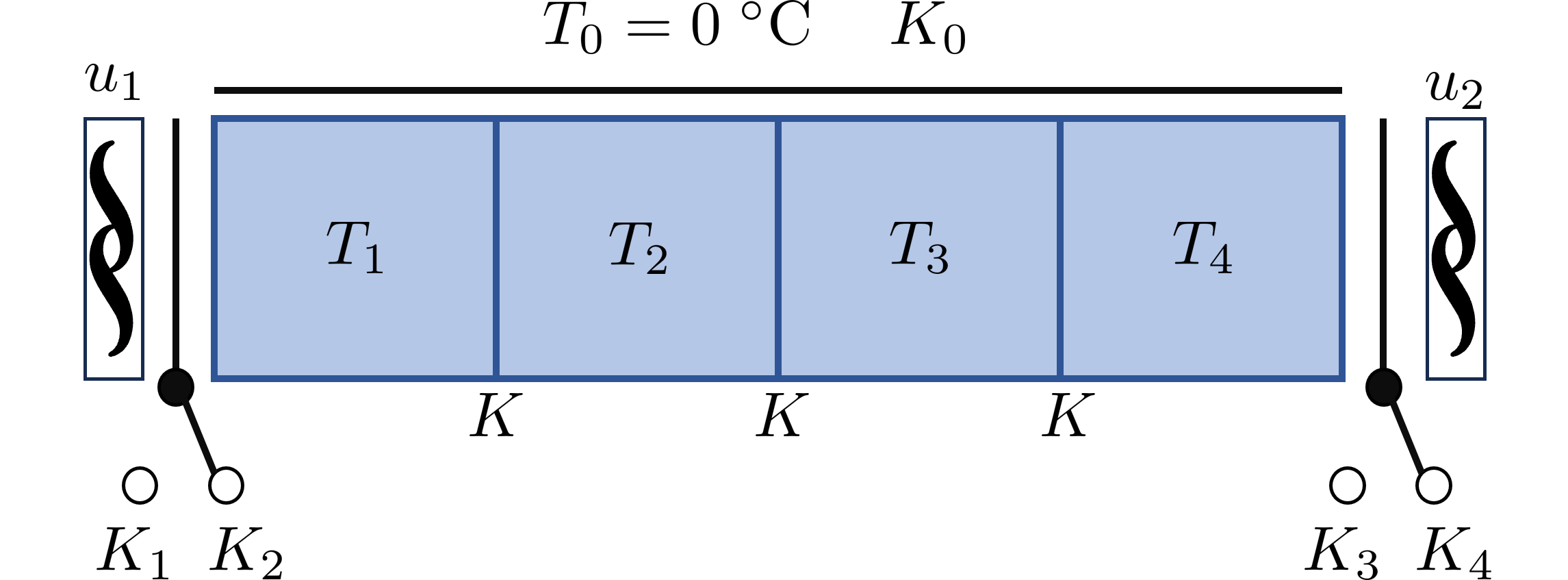}
	\caption{Diagram of the building.
	\label{fig_room}}
\end{figure}

The ambient temperature $T_0$ outside the building is assumed to be constantly zero, i.e., $T_0 = 0\;^\circ\text{C}$. Heat flows occur between each two adjacent rooms as well as between the external environment and each room. These thermal interactions are modeled using linear equations. For example, the temperature dynamics in Room 2 are given by $\dot T_2 = K (T_1 - T_2) + K(T_3 - T2) + K_0 (T_0 - T_2)$, where $K$ and $K_0$ are effective heat-transfer coefficients determined by the corresponding thermal resistances and capacitances. Room 1 and Room 4 are each equipped with a temperature controller ($u_1$ and $u_2$) to regulate the temperature in the respective room. We assume that the heaters are powerful enough such that they realize the demanded temperature values $u_1$ and $u_2$ immediately. The heat flow between a room and its controller is modeled by a piecewise linear equation. In particular, when the temperature difference $|u_1 - T_1|$ exceeds a threshold of $5\;^\circ\text{C}$, the system automatically switches to a mode with a larger coefficient $K_2 > K_1$ to facilitate better heat conduction. To summarize, the system admits the following state space equation:
\begin{equation}\label{eq_room}
\small
	\begin{aligned}
	\dot T_1 &=\begin{cases}
		K(T_2 \!- \!T_1) \!+\! K_1 (u_1 \!- \!T_1) \!+\! K_0 (T_0 \!-\! T_1),\mathrm{if}\; |u_1 \!- \!T_1| \leq 5\\
		K(T_2 - T_1) + K_2 (u_1 - T_1) + K_0 (T_0 - T_1),\mathrm{otherwise}
	\end{cases}\nonumber\\
	\dot T_2 &= K (T_1 - T_2) + K(T_3 - T_2) + K_0 (T_0 - T_2)\nonumber\\
	\dot T_3 &= K (T_2 - T_3) + K(T_4 - T_3) + K_0 (T_0 - T_3)\nonumber\\
	\dot T_4 &=\begin{cases}
		K(T_3 \!- \!T_4) \!+ \!K_3 (u_2 \!-\! T_4)\! +\! K_0 (T_0 \!- \!T_4),\mathrm{if}\; |u_2 \!-\! T_4| \leq 5\\
		K(T_3 - T_4) + K_4 (u_2 - T_4) + K_0 (T_0 - T_4),\mathrm{otherwise},
	\end{cases}
\end{aligned}
\end{equation}
where $K = 0.0035\;\mathrm{s}^{-1}$, $K_0 = 0.001\;\mathrm{s}^{-1}$, $K_1 = 0.01\;\mathrm{s}^{-1}$, $K_2 = 0.02\;\mathrm{s}^{-1}$, $K_3 = 0.008\;\mathrm{s}^{-1}$, and $K_4 = 0.016\;\mathrm{s}^{-1}$. This PWA system can be rewritten as the standard form of \eqref{pwa_form}, with a total of 9 polyhedral regions. 

The control objective is to adjust the temperatures of Room 2 and Room 3 to $20\;^\circ\text{C}$, while keeping the temperatures of Room 1 and Room 4 no higher than $26\;^\circ\text{C}$. The input constraints are given by $u_1,\;u_2\in [-10,\;35]$. The initial state is $[5\; 20\; 10\; 15]^T$. In the simulation, the system is discretized using the forward Euler method with a sampling time of 1 second. 

The reference controller is designed based on backstepping \cite[Chapter 14.3]{khalil2002nonlinear}. In particular, to make $T_2$ and $T_3$ track $20\;^\circ\text{C}$, a virtual controller $v$ for $T_1$ and $T_4$ is designed as 
\begin{align*}
	v=K_{23}^u\left( \begin{bmatrix} T_2 \\ T_3 \end{bmatrix}- \begin{bmatrix} 20 \\ 20 \end{bmatrix}\right)-B_{23}^{-1}A_{23} \begin{bmatrix} 20 \\ 20 \end{bmatrix}
\end{align*}
where $A_{23}$ and $B_{23}$ are the system matrices of the linear sub-equations for $\dot T_2$ and $\dot T_3$, where $T_2$ and $T_3$ are the states and $T_1$ and $T_4$ are the inputs. The feedback gain $K^u_{23}$ is chosen as the LQR gain for the linear subsystem with $Q = I_2$ and $R=I_2$. The second term of $v$ is a feedforward term. Then, the real reference input is determined by
\begin{align*}
	\pi_{\mathrm{r}}(x) = K_{14}^u \left( \begin{bmatrix} T_1 \\ T_4 \end{bmatrix}- v\right) - B_{14}^{-1}A_{14} v- \begin{bmatrix}T_2 K/K1 \\T_3K/K3\end{bmatrix}
\end{align*}
where $A_{14}$ and $B_{14}$ are the system matrices of the linear subsystem for $T_1$ and $T_4$, corresponding to the region where $|u_1 - T_1| \leq 5$ and $|u_2 - T_4| \leq 5$, with $T_2$ and $T_3$ treated as external disturbances. The feedback gain $K^u_{14}$ is designed similarly to $K^u_{23}$. The second term of $\pi_{\mathrm{r}}$ is a feedforward term, while the third term compensates for the external disturbances.

{\color{blue}\begin{lemma}[Verifying hypotheses of Corollary \ref{corollary_forward}]\label{lem:verifying corollary}
Consider the backup closed loop $\dot x=f_\mathrm{PWA,b}(x)$ of the building
example, and the backup CBF $h_\mathrm{b}(x)=\min\{24-T_1,\,24-T_4\}$. Let $
\mathcal{O}=\big\{x\in\mathbb{R}^4 :\ -4.5\le T_i\ (i=1,\dots,4),\
T_1,T_4\le 26,\ T_2,T_3\le 21\big\}$. Then
\begin{enumerate}
\item $\mathcal{O}$ is forward invariant for $\dot x=f_\mathrm{PWA,b}(x)$;
\item there exists an $\alpha_b \in \mathcal{K}_\infty$ such that $\partial h_\mathrm{b}f_\mathrm{PWA,b}(x)\ge-\alpha_\mathrm{b}(h_\mathrm{b}(x))$
for all $x\in S_\mathrm{b}\cap \mathcal{O}$ and all
$\partial h_\mathrm{b}\in\partial_\mathrm{C}h_\mathrm{b}(x)$.
\end{enumerate}
\end{lemma}
\begin{proof} See Appendix \ref{appendix:proof of verifying corollary}.
\end{proof}}

\textcolor{blue}{The temperature-control example compares two controllers:
\begin{itemize}
    \item The policy $\pi_\mathrm{safe}$ obtained using the proposed all-elements PSF by solving an MIQP.
    \item The approximate explicit safe policy $\hat{\pi}_\mathrm{safe}$ computed using \eqref{approximated_policy}, \eqref{lamda expression}, and \eqref{lamda ij}.
\end{itemize}}

Fig. \ref{fig_time_traj2} plots the time responses of the closed-loop system controlled by $\hat\pi_\mathrm{safe}$ and $\pi_\mathrm{safe}$. The PSF \eqref{safety filter pwa} is implemented over a prediction horizon of $T=4$ seconds, which is discretized into $N=40$ uniformly spaced time steps. This corresponds to a discretization interval of 0.1 seconds. From the top-left and bottom-right sub-figures, it is observed that both controllers satisfy the safety requirements ($T_1, T_4 \leq 26 \;^\circ\text{C}$). On the other hand, the top-right and bottom-left sub-figures show that both $\hat\pi_\mathrm{safe}$ and $\pi_\mathrm{safe}$ with $N= 40$ drive the temperatures $T_2$ and $T_3$ closer to the target temperature ($20 \;^\circ\text{C}$), compared to $\hat\pi_\mathrm{safe}$ and $\pi_\mathrm{safe}$ with $N= 0$ (i.e., policies derived from a standard CBF-based safety filter without predictions). Furthermore, consistent with intuition, the approximate explicit policy $\hat\pi_\mathrm{safe}$ has a slightly worse tracking performance than the exact one $\pi_\mathrm{safe}$, as evidenced by the red and blue curves in the top-right and bottom-left sub-figures.
\begin{figure}
	\centering
	\includegraphics[width=240pt,clip]{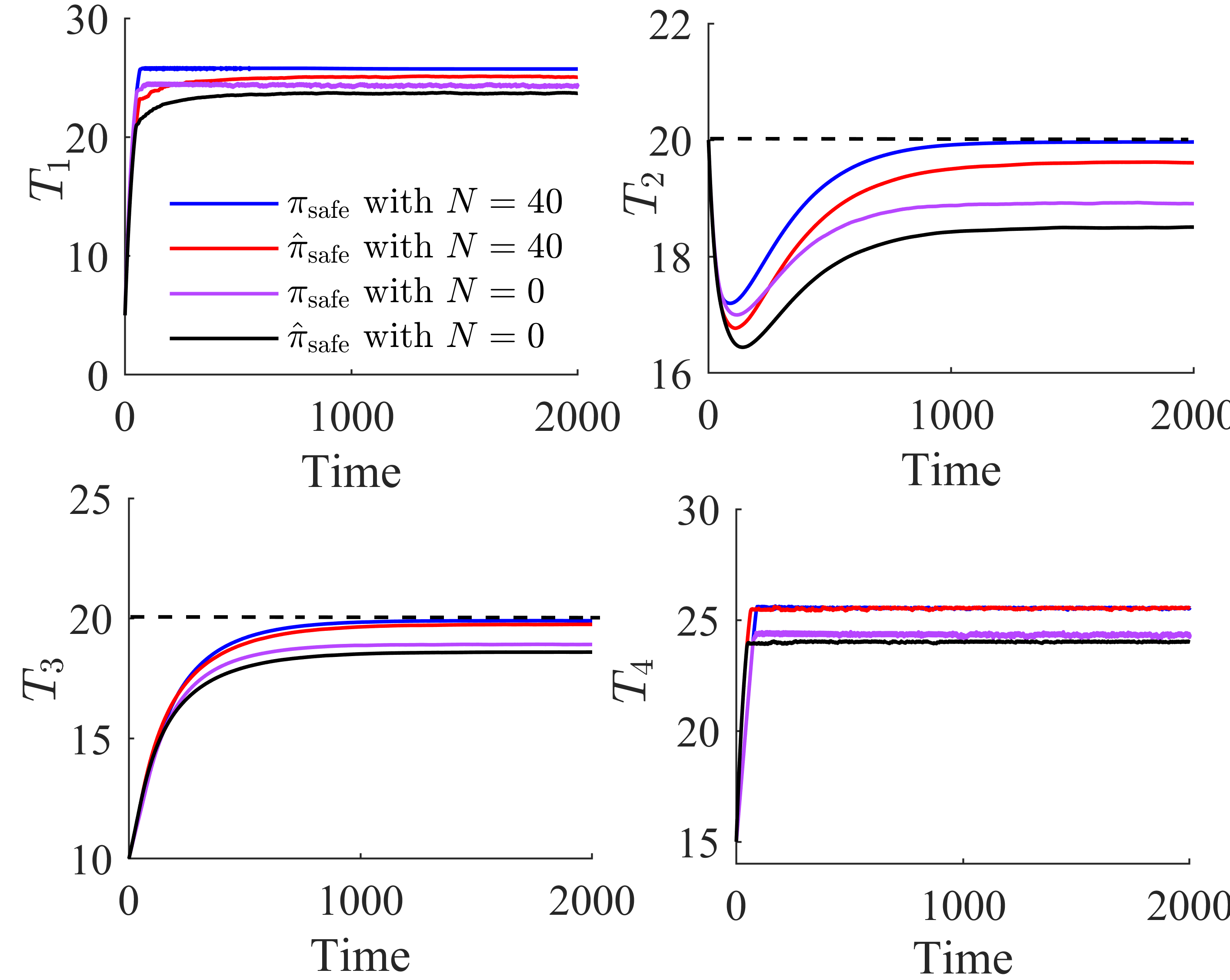}
	\caption{Time responses of the closed-loop system controlled by $\hat\pi_\mathrm{safe}$ and $\pi_\mathrm{safe}$. The black dashed lines represent the target temperature. The policy with $N=40$ is derived from \eqref{safety filter pwa} with the prediction horizon $T =4$ seconds and discretized into $N=40$ uniformly spaced time steps. \textcolor{blue}{The worst-case constraint value, $\min_{t\in[0,2000]} h_X(x(t))$, is 0.1638 and 1.4568 for $\pi_\mathrm{safe}$ with $N=40$ and $N=0$, respectively, and is 0.3979 and 1.9415 for $\hat\pi_\mathrm{safe}$ with $N=40$ and $N=0$, respectively. This implies safety throughout the simulation.} The policy with $N=0$ is derived from a standard CBF-based safety filter without any predictions. Both $\hat\pi_\mathrm{safe}$ and $\pi_\mathrm{safe}$ with $N=40$ drive the temperatures $T_2$ and $T_3$ closer to the target temperature ($20 \;^\circ\text{C}$), compared to $\hat\pi_\mathrm{safe}$ and $\pi_\mathrm{safe}$ with $N= 0$. As expected, the approximate explicit policy $\hat\pi_\mathrm{safe}$ exhibits a slightly worse tracking performance compared to the exact policy $\pi_\mathrm{safe}$.}
	\label{fig_time_traj2}
\end{figure}

In Fig. \ref{fig_compare}, we compare the tracking performance and average CPU time of the controller $\pi_\mathrm{safe}$ and its explicit approximation $\hat\pi_\mathrm{safe}$ for different prediction horizons. The tracking performance is defined as the summation of $(T_2 - 20)^2 + (T_3 - 20)^2$ over 2000 seconds. To implement $\pi_\mathrm{safe}$, according to Section \ref{section implementation}, an MIQP problem needs to be solved online. \textcolor{blue}{To ensure a fair comparison of CPU times, we compute the big-M bounds by solving several linear programs \cite{bemporad1999control}, which provide the tightest bounds that yield an equivalent mixed-integer encoding. We use Gurobi 10.0.3 \cite{gurobi2018gurobi} to solve the MIQPs in MATLAB R2021a on a laptop equipped with an AMD Ryzen 7 5800H CPU running at 3.20 GHz. The CPU time is measured with the MATLAB commands \texttt{tic} and \texttt{toc}, applied around the computation of the
control signal at each time step, while the time for updating (simulating) the system is excluded.} The results show that the approximate controller $\hat\pi_\mathrm{safe}$ retains a comparable tracking performance to the exact one $\pi_\mathrm{safe}$, and the average CPU time for $\hat\pi_\mathrm{safe}$ (red line) is 1–2 orders of magnitude lower than that for the exact policy $\pi_\mathrm{safe}$ (blue line) across all prediction horizons, highlighting the practicality of the approximation in real-time or resource-constrained settings. Meanwhile, increasing the prediction horizon leads to reduced tracking errors for both $\hat\pi_\mathrm{safe}$ and $\pi_\mathrm{safe}$, while their CPU times grow approximately linearly. This demonstrates that incorporating predictions can effectively reduce the conservatism of the safety filter, compared to a standard CBF-based approach without predictions.

\begin{figure}
	\subfloat[Accumative tracking error.]{\includegraphics[width=120pt,clip]{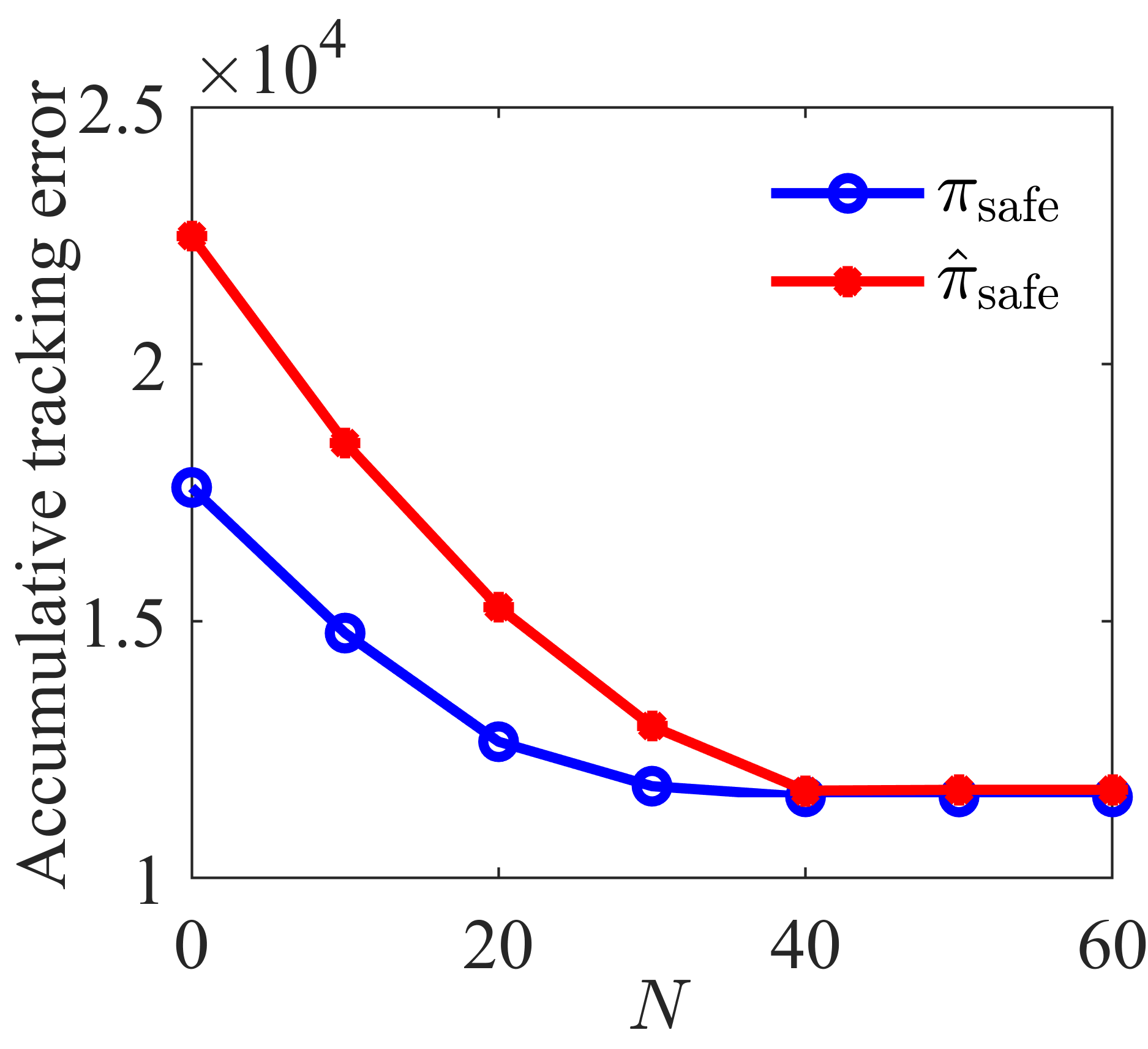}}
	\hfil
	\subfloat[Average CPU time (log scale).]{\includegraphics[width=120pt,clip]{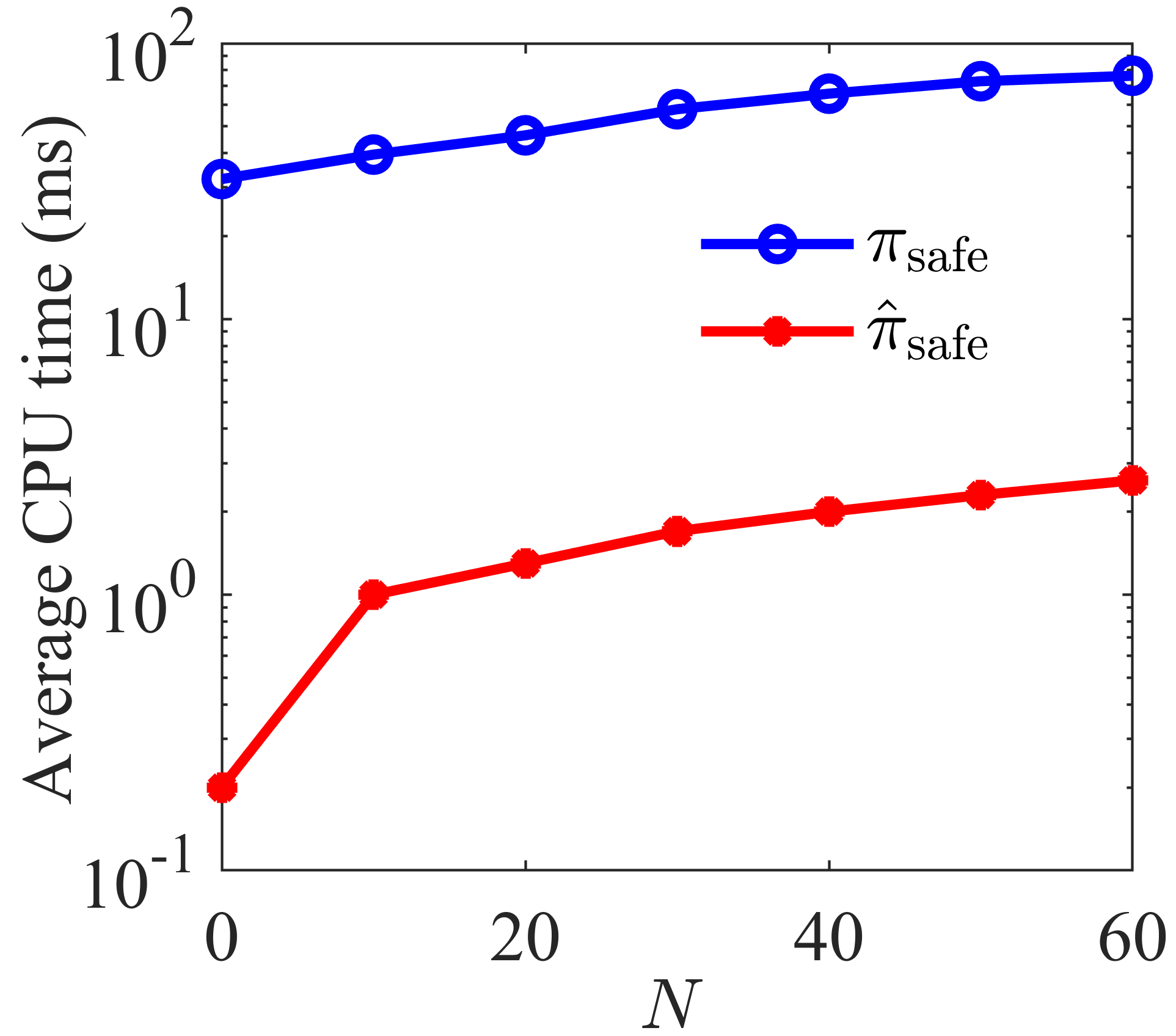}}
	\caption{Tracking performance and average CPU time (log scale) for different prediction horizons for the exact policy $\pi_\mathrm{safe}$ and the approximate policy $\hat\pi_\mathrm{safe}$. The tracking errors (defined as the summation of $(T_2 - 20)^2 + (T_3 - 20)^2$ over 2000 seconds) are comparable, but the approximate policy significantly reduces the computational cost by up to two orders of magnitude across all horizons.}
	\label{fig_compare}
\end{figure}

\section{Conclusions and future work}
This paper has developed a novel predictive safety filter (PSF) to achieve less conservative control of constrained piecewise affine (PWA) systems compared to standard safety filters that use control barrier functions (CBFs) or projections onto invariant sets. We have extended the predictive CBF method to handle the non-smoothness occurring in PWA dynamics, in the state constraint, as well as in the CBF. We have further derived an explicit approximation of the solution to the PSF. Such an approximation can be efficiently computed through basic arithmetic and min/max operations. A case study on an inverted pendulum interacting with an elastic wall has shown improved safety performance compared to classical PSFs. Furthermore, simulations on temperature control of a building have demonstrated the computational advantages of the approximate solution. 

In the future, we will consider designing PSFs for general Lipschitz continuous nonlinear systems and exploring their explicit approximation. The PSF method can be viewed as an online version of HJ reachability. We believe that unifying these two methods is a meaningful future research direction. Besides, another interesting topic is to consider probabilistic safety constraints under model uncertainties.

\appendix
\subsection{Proof of Lemma \ref{lemma_invariance}}\label{appendix_lemma_invariance}
For the convenience of using the comparison theorem \cite[Lemma 3.4]{khalil2002nonlinear}, we define $\bar h = -h$. Since $\bar h$ is not always differentiable, we consider the generalized directional derivative \cite{clarke1975generalized}
\begin{equation}\label{directional}
	\bar h^{\circ}(x ; v):=\limsup_{w \rightarrow 0 , \delta \downarrow 0}\frac{\bar h(x+w+\delta v)-\bar h(x+w)} { \delta}
\end{equation}
for any $v \in \mathbb{R}^n$. The generalized directional derivative is well-defined for Lipschitz continuous functions. According to \cite[Proposition 1.4]{clarke1975generalized}, we have 
\begin{equation*}
	\bar h^{\circ}(x ; v) = \max_{\partial \bar h \in \partial_\mathrm{C} \bar h(x)} \{\partial \bar h v \}.
\end{equation*}
By specifying $v = f(x,\pi(x))$ and noting that $\partial_\mathrm{C} \bar h(x) = - \partial_\mathrm{C}  h(x)$ and that \eqref{cbf_condition} holds, we have
\begin{align}\label{directional_decrease}
	\bar h^{\circ}(x ; f(x,\pi(x))) =& \max_{\partial \bar h \in \partial_\mathrm{C} \bar h(x)} \{\partial \bar h f(x,\pi(x)) \} \leq \alpha (-\bar h(x)),\nonumber\\
	&\forall x \in \{x \in \mathbb{R}^n| \bar h(x) \leq  0\}.
\end{align}

Consider the upper right-hand time derivative of $\bar h$ along $\dot x = f(x,\pi(x))$, which is defined by
\begin{equation}\label{dini}
	\mathrm{D}^+ \bar{h}(x_0,t) : = \limsup_{\tau \downarrow 0 } \frac{\bar h(\phi(x_0,t+\tau))-\bar h(\phi(x_0,t))} {\tau}.
\end{equation}
Since $\phi$ is the solution of the closed-loop system $\dot x = f(x,\pi(x))$, we have the relation $\phi(x_0,t) + \limsup_{w \rightarrow 0 , \delta \downarrow 0} w+ \delta f(\phi(x_0,t),\pi(\phi(x_0,t))) = \limsup_{\tau \downarrow 0 }\phi(x_0,t+\tau)$. As a result, by comparing \eqref{directional} and \eqref{dini} and noting that \eqref{directional_decrease} holds, we obtain
\begin{align*}
	\mathrm{D}^+ \bar{h}(x_0,t) &= 	\bar h^{\circ}(\phi(x_0,t) , f(\phi(x_0,t),\pi(\phi(x_0,t))))\nonumber\\
	&\leq \alpha(-\bar h(\phi(x_0,t))),\;\forall t \geq 0 
\end{align*}
if $x_0 \in \{x \in \mathbb{R}^n \mid \bar{h}(x) \leq 0\}$. Because $\alpha$ is Lipschitz continuous, directly applying the comparison theorem \cite[Lemma 3.4]{khalil2002nonlinear} yields $\bar{h}(\phi(x_0,t)) \leq 0,\;\forall t \geq 0 $, which consequently proves Lemma \ref{lemma_invariance}.

\subsection{Proof of Proposition \ref{proposition1}}\label{appendix_proposition1}
The equivalence between (3) and (4) follows from the definitions of backward and forward reachable sets. The equivalence between (1) and (2) stems from Definition \ref{Aumann}.

If $\Phi_\mathrm{for}(\bar x_0) \cap \mathcal{C} \ne \emptyset$, $\phi_\mathrm{b}(\bar x_0,\tau)$ is on the boundary of multiple polyhedra of the PWA system \eqref{closed_loop_pwa} during $[\tau_k,\; \tau_{k+1})$, where $\{(\tau_0, I_0),\;(\tau_1, I_1),\;...,\;(\tau_M, I_M)\}$ constitutes the time-stamped switching sequence of the solution $\phi_\mathrm{b}(\bar x_0,\tau)$. To see why this is true, note that if $\Phi_\mathrm{for}(\bar x_0) \cap \mathcal{C} \ne \emptyset$, there exists at least one instant $\tau_k \geq 0,\; k\in \mathbb{N}_M$ such that $\phi_\mathrm{b}(\bar x_0,\tau_k) \in \mathcal{C}$. According to the definition of $\mathcal{C}$, $\bar{g}^{(q)}_{I_k}(x(\tau_k)) = g_{I_k} D^{q-1}_i (D_i x(\tau_k) + d_i)=0,\;\forall q \in \mathbb{N}^+_n$. Using the Cayley–Hamilton Theorem \cite{decell1965application}, we know that $\bar{g}^{(q)}_{I_k}(x(\tau_k)) =0,\;\forall q \in \mathbb{N}$, meaning that any order time derivative of $\bar{g}_{I_k}$ vanishes when $\tau = \tau_k$. Consequently, $\bar{g}_{I_k}$ remains zero during $[\tau_k, \tau_{k+1})$, $\tau_{k+1} >\tau_k$, i.e., the set $\Gamma_I$ is an \emph{invariant manifold} \cite[Chapter 8.1]{khalil2002nonlinear}, so the trajectory must stay in $\Gamma_I$ over the time interval $[\tau_k, \tau_{k+1})$. We refer to all such intervals as \emph{critical periods}. During these periods, $\phi_\mathrm{b}$ resides at the intersection of the closure of multiple regions. 

\emph{(3) $\Rightarrow$ (2):} If $\Phi_\mathrm{for}(\bar x_0) \cap \mathcal{C} = \emptyset$, we can distinguish two cases. In the first case, there is no critical period. This means that for any switching instant $\tau_k$, the trajectory crosses the hyperplane $\{x \in \mathbb{R}^n | \bar{g}_{I_k\cup I_{k+1}}(x)=0\}$ from $\mathcal{R}_{I_k}$ and immediately enters the interior of one of the intersecting polyhedra $\mathcal{R}_{I_{k+1}}$. Here, all $I_k$ have only one element. When $\tau \in (\tau_k,\; \tau_{k+1})$, we know that $\partial f_\mathrm{PWA,b} (\phi_{\mathrm{b}}(\bar x_0,\tau))$ is uniquely determined by $ D_{I_k}$. As a result, according to \eqref{generalized sensitivity computation}, $Q_\mathrm{A}(\bar x_0,\tau)$ is single-valued. 

In the second case of $\Phi_\mathrm{for}(\bar x_0) \cap \mathcal{C} = \emptyset$, there still exists one or more critical periods during which $\bar{g}^{(q)}_{I_k} =0,\;\forall q \in \mathbb{N}^+_n$. However, it must hold that $D_{i_k} = D_{j_k}$ for every two adjacent polyhedra $\mathcal{R}_{i_k}$ and $\mathcal{R}_{j_k}$, $i_k,j_k \in I_k$. In this case $Q_\mathrm{A}(\bar x_0,\tau)$ is still single-valued.  

\emph{(2) $\Rightarrow$ (3):} We prove this implication by showing $\neg$ (3) $\Rightarrow$ $\neg$ (2). If $\Phi_\mathrm{for}(\bar x_0) \cap \mathcal{C} \ne \emptyset$, according to the definition of $\mathcal{C}$, during each critical period $[\tau_k, \tau_{k+1})$, $\phi_\mathrm{b}(\bar x_0,\tau)$ lies at the intersection of two or more adjacent polyhedra. As a result, the integral term of \eqref{generalized sensitivity computation} admits at least two different forms $D_{i_k} Q(\bar x_0,s)$ and $D_{j_k} Q(\bar x_0,s)$, $i_k,j_k \in I_k$ from $s=\tau_k$ to $s= \tau_{k+1}$. By holding the integral term during the other periods $[\tau_0,\tau_1),...,[\tau_{k-1},\tau_k),\;[\tau_{k+1},\tau_{k+2}),...$ fixed, we know that $Q_\mathrm{A}(\bar x_0,\tau)$ takes at least two different values for all $\tau>\tau_k$.

In addition, if $\mathcal{C}$ is empty, then it follows from the above discussions that $\phi_\mathrm{b}(x_0,\tau)$ has a global classical Jacobian w.r.t. $x_0$.

\subsection{Proof of Proposition \ref{proposition_equivalence}}\label{appendix_proposition_equivalence}

If $\phi_\mathrm{b}(x_0,\tau)$ is differentiable w.r.t. $x_0$ at $x_0 = \bar x_0$, it is obvious that $\partial_\mathrm{C} \phi_\mathrm{b}(\bar{x}_0,\tau) = Q_\mathrm{A} (\bar x_0,\tau) = \left.\frac{\partial \phi_{\mathrm{b}}\left(x_0, \tau\right)}{\partial x_0}\right|_{x_0=\bar{x}_0}$. In the remainder of the proof, we consider the case when $Q_\mathrm{A}(\bar{x}_0,\tau)$ has multiple values.

First, we prove $\partial_\mathrm{C} \phi_\mathrm{b} \subseteq \mathrm{conv}(Q_\mathrm{A})$. For any $\bar x_0$ and $\tau$ such that $\phi_\mathrm{b}(x_0,\tau)$ is not differentiable w.r.t. $x_0$ at $x_0 = \bar x_0$, consider any sequence $\{y_j\}^\infty_{j=0}$ satisfying $\lim_{j \to \infty} y_j = \bar x_0$ and that $\phi_\mathrm{b}(x_0,\tau)$ is differentiable w.r.t. $x_0$ at $y_j$. Since $y_j$ converges to $\bar x_0$, there exists a sufficiently large $\bar j \in \mathbb{N}^+$ such that for any $j\geq \bar{j}$, $\phi_\mathrm{b}(y_j,\tau)$ admits the same switching sequence $\{i_0,\;i_1,\;...,i_M\}$ satisfying $\{i_0,\;i_1,\;...,i_M\} \in I(\bar x_0)$, where $I(\bar x_0)$ is the Cartesian product of the switching sequence of $\phi_\mathrm{b}(\bar x_0,\tau)$.

According to Definition \ref{Clarke definition}, $\lim_{j \to \infty} J_{\phi_\mathrm{b}}(y_j,\tau) \in \partial_\mathrm{C}  \phi_\mathrm{b}(\bar x_0,\tau)$, where $J_{\phi_\mathrm{b}}$ refers to the classical Jacobian and has the expression
\begin{align}\label{recursive 2}
	J_{\phi_\mathrm{b}}(y_j,\tau) \! =\!  \left\{ \begin{gathered}
		{e^{{D_{{i_0}}}\tau}},\tau \in [{\tau_0},{\tau_1}) \hfill \\
		{e^{{D_{{i_1}}}(\tau - {\tau_1})}}J_{\phi_\mathrm{b}}(y_j,{\tau_1}),\tau \in [{\tau_1},{\tau_2}) \hfill \\
		... \hfill \\
		{e^{{D_{{i_{M - 1}}}}(\tau - {\tau_{M - 1}})}}J_{\phi_\mathrm{b}}(y_j,{\tau_{M - 1}}),\tau\! \in\! [{\tau_{M - 1}},{\tau_M}).  
	\end{gathered}  \right.
\end{align}
By comparing \eqref{recursive} and \eqref{recursive 2} and noting that the switching time $\tau_k$ is continuous to $y_j$, we can prove by induction that $\lim_{j \to \infty} J_{\phi_\mathrm{b}}(y_j,\tau) = Q_{(i_0,i_1,...,i_M)}(\bar x_0,\tau)\in Q_\mathrm{A}(\bar x_0, \tau)$. Here, the subscript ${(i_0,i_1,...,i_M)}$ means that $Q_{(i_0,i_1,...,i_M)}$ corresponds to the element of $Q_\mathrm{A}$ that is determined by the switching sequence ${(i_0,i_1,...,i_M)}$. The containment of $\partial_\mathrm{C} \phi_\mathrm{b}$ in $\mathrm{conv}(Q_\mathrm{A})$ follows from the arbitrariness of the sequence $\{y_j\}^\infty_{j=0}$.

Conversely, we prove $ \mathrm{conv}(Q_\mathrm{A}) \subseteq \partial_\mathrm{C} \phi_\mathrm{b} $. Since we have proved that $\partial_\mathrm{C} \phi_\mathrm{b} \subseteq \mathrm{conv}(Q_\mathrm{A})$, according to Assumption \ref{assumption_invert}, any element of $\partial_\mathrm{C} \phi_\mathrm{b}$ is invertible. As a result, by virtue of the inverse function theorem for generalized Clarke derivatives \cite[Theorem 1]{clarke1976inverse}, we know that there exist a Lipschitz continuous function $\phi^{-1}_{\mathrm{b}}(\cdot,\tau): \mathbb{R}^n \to \mathbb{R}^n $, a neighborhood $Y$ of $\bar x_0$, and a neighborhood $Z$ of $\phi_{\mathrm{b}}(\bar x_0,\tau)$ such that $\phi^{-1}_{\mathrm{b}}(\phi_{\mathrm{b}}(y,\tau)) = y,\;\forall y \in Y$ and $\phi_{\mathrm{b}}(\phi^{-1}_{\mathrm{b}}(z,\tau)) = z,\;\forall z \in Z$. 

Then, for any $\bar x_0$ and any $\tau$ such that $\phi_\mathrm{b}(x_0, \tau)$ is not differentiable w.r.t. $x_0$ at $x_0 = \bar x_0$, consider any element $Q_{(i_0,i_1,...,i_M)}(\bar x_0,\tau )$ in $Q_\mathrm{A}(\bar x_0,\tau )$. Note that $Q_{(i_0,i_1,...,i_M)}(\bar x_0,\tau )$ can be expressed by \eqref{recursive}. Since $\phi_{\mathrm{b}}$ is differentiable almost everywhere, we immediately know that $\Phi_{\text {back }}(\mathcal{C})$ has the dimension strictly lower than $n$. Therefore, we can always find a sequence $\{z_j\}^\infty_{j=1}$ satisfying $z_j \in Z,\;\forall j \in \mathbb{N}^+$, $z_j \notin \Phi_{\text {back }}(\mathcal{C}),\;\forall j \in \mathbb{N}^+$, and $\lim_{j \to \infty} z_j = \phi_{\mathrm{b}}(\bar x_0,\tau)$. By substituting $z_j$ to the inverse function $\phi^{-1}_{\mathrm{b}}$, we get a sequence $y_j = \phi^{-1}_{\mathrm{b}}(z_j,\tau),\; j =1,2,..$. Thanks to the Lipschitz continuity of the inverse function, the sequence $\{y_j\}^\infty_{j=0}$ satisfies $\phi_{\mathrm{b}}(y_j,\tau) = z_j$, $\lim_{j \to \infty} y_j = \lim_{j \to \infty} \phi^{-1}_{\mathrm{b}}(z_j,\tau) = \bar x_0$, and $y_j \notin  \Phi_{\text {back }}(\mathcal{C})$. As $y_j$ converges to $\bar x_0$, there exists a $\bar j \in \mathbb{N}^+$ such that for any $j\geq \bar{j}$, $\phi_\mathrm{b}(y_j,\tau)$ admits the same switching sequence $\{i_0,\;i_1,\;...,i_M\} \in I(\bar x_0)$. As a result, $\lim_{j \to \infty} J_{\phi_\mathrm{b}}(y_j,\tau) = Q_{(i_0,i_1,...,i_M)}(\bar x_0,\tau)$. It follows from Definition \ref{Clarke definition} that $Q_{(i_0,i_1,...,i_M)}(\bar x_0,\tau) \in \partial_\mathrm{C} \phi_\mathrm{b} (\bar x_0,\tau)$. This implies $\mathrm{conv}(Q_\mathrm{A} (\bar x_0,\tau)) \subseteq \partial_\mathrm{C} \phi_\mathrm{b}(\bar x_0,\tau)$, which completes the proof of Proposition \ref{proposition_equivalence}.

\subsection{Proof of Proposition \ref{proposition_explicit}}\label{appendix_proposition_explicit}
Let $\lambda^+(x)$ denote the right-hand side of \eqref{lamda ij}. The proof is divided into two parts: proving the feasibility of $\lambda^+(x)$ and the optimality of $\lambda^+(x)$.

(i) \emph{Feasibility of $\lambda^+(x)$.} For any $x \in \Phi_{\text {const}}\left(S_{\mathrm{b}}, T\right)$, since $\pi_\mathrm{b}(x)$ is feasible for \eqref{safety filter pwa}, $0$ is feasible for \eqref{lamda compact} as well. Therefore, we know that $\exists i \in I_u(x)$ such that $\omega_{i,j}(x,\tau) \geq 0,\;\forall \tau \in [0,T]$ and $\forall j \in \mathbb{N}^+_{\mathrm{dim}(\eta_i)}$. Due to the maximization operator in \eqref{lamda expression}, it further holds that $\lambda^+(x) \geq 0$.

Suppose without loss of generality that $\lambda^+(x) = \lambda_{i^+,j^+}(x,\tau^+)$ for a certain $i^+\in I_u(x)$, $j^+ \in \mathbb{N}^+_{\mathrm{dim}(\eta_{i^+})}$, and $\tau^+ \in[0,T]$. It follows from the two minimization operators in \eqref{lamda expression} that 
\begin{equation}\label{relation}
	\lambda^+(x) \leq \lambda_{i^+,j}(x,\tau),\;\forall j \in \mathbb{N}^+_{\mathrm{dim}(\eta_{i^+})} \text{ and } \forall \tau \in [0,T].
\end{equation}
Then, for any $j \in \mathbb{N}^+_{\mathrm{dim}(\eta_{i^+})}$ and any $\tau \in [0,T]$, based on \eqref{lamda ij} we can distinguish the following four cases:

Case (1): $\eta_{i^+,j}(x,\tau) >0$. In this case, $\lambda_{i^+,j}(x,\tau)  = \frac{\omega_{i^+,j}(x,\tau)}{\eta_{i^+,j}(x,\tau)}$. Directly applying \eqref{relation} yields 
\begin{equation}\label{feasibility}
	\eta_{i^+,j}(x,\tau) \lambda^+(x)\leq \omega_{i^+,j}(x,\tau).
\end{equation}

Case (2): $\eta_{i^+,j}(x,\tau) =0$ and $\omega_{i^+,j}(x,\tau) \geq0$. In this case, \eqref{feasibility} obviously holds.

Case (3): $\eta_{i^+,j}(x,\tau) =0$ and $\omega_{i^+,j}(x,\tau) <0$. From \eqref{lamda ij}, we know that $\lambda_{i^+,j}(x,\tau) =0$, which further implies $\lambda^+(x) =0$ because of \eqref{relation} and $\lambda^+(x) \geq 0$. Consequently, the feasibility of $\lambda^+(x) =0$ results from the feasibility of $\pi_\mathrm{b}(x)$.

Case (4): $\eta_{i^+,j}(x,\tau) <0$. In this case, we only need to focus on the situation when $\lambda^+(x)>0$ because if $\lambda^+(x)=0$ then the feasibility of $\lambda^+(x)$ follows from the feasibility of $\pi_\mathrm{b}(x)$. If $\lambda^+(x)>0$, then by defining $\lambda_{i^+}(x): = \min_{\tau' \in[0,T], j' \in \mathbb{N}^+_{\mathrm{dim}(\eta_{i^+})}, \lambda_{i^+,j'}(x,\tau') \in (0,1)}  \lambda_{i^+,j'} (x,\tau') $, we have 
\begin{equation}\label{40}
	\omega_{i^+,j}(x,\tau) -\eta_{i^+,j}(x,\tau) \lambda_{i^+}(x)\geq 0.
\end{equation}
Noticing the relation $\lambda^+(x) = \lambda_{i^+,j^+}(x,\tau^+) \geq \lambda_{i^+}(x)$, we derive \eqref{feasibility} from \eqref{40}.

Combining the four cases above and noticing the arbitrariness of $j$ and $\tau$, we can conclude that $\lambda^+(x)$ satisfies the constraint of \eqref{lamda compact} for $i = i^+$. This proves the feasibility of $\lambda^+(x)$.

(ii) \emph{Optimality of $\lambda^+(x)$.} If $\lambda^+(x) =1$, the optimality follows directly from the constraint $\lambda \leq 1$. In the remainder of the proof, we consider $\lambda^+(x) <1$ and prove that $\lambda^+(x) +\Delta_\lambda$ is infeasible for problem \eqref{lamda compact} for any $\Delta_\lambda \in\left(0,1-\lambda^{+}(x)\right] $. Since the objective of the optimization problem \eqref{lamda} is to maximize $\lambda$, the infeasibility of $\lambda^+(x) +\Delta_\lambda$ implies the optimality of $\lambda^+(x)$.

For any $i \in I_u(x)$, let $\lambda_i(x) = \lambda_{i,j^*}(x, \tau^*):= \min_{\tau \in[0,T]} \min_{j \in \mathbb{N}^+_{\mathrm{dim}(\eta_i)}} \lambda_{i,j}(x,\tau)$, and then we get the relation
\begin{equation}\label{relation2}
	\max_{i\in  I_u(x)} \lambda_i(x) = \lambda^+(x) <1.
\end{equation}

Analogous to the proof of the feasibility, for any $i \in I_u(x)$, we consider the following four cases:

Case (1): $\eta_{i,j^*}(x,\tau^*) >0$. In this case, $\lambda_{i,j^*}(x,\tau^*)  = \frac{\omega_{i,j^*}(x,\tau^*)}{\eta_{i,j^*}(x,\tau^*)}$. Therefore, by \eqref{relation2} we have
\begin{align*}
	&\eta_{i,j^*}(x,\tau^*) (\lambda^+(x) + \Delta_\lambda) \\
	\geq& \underbrace{\eta_{i,j^*}(x,\tau^*) \lambda_{i,j^*}(x,\tau^*)}_{=\omega_{i,j^*}(x,\tau^*)} + \underbrace{\eta_{i,j^*}(x,\tau^*)  \Delta_\lambda}_{>0}\\
	>& \omega_{i,j^*}(x,\tau^*),
\end{align*}
which shows the infeasibility of $\lambda^+(x) +\Delta_\lambda$ regarding the $j^*$-th component of the constraint of \eqref{lamda compact}. 

Case (2): $\eta_{i,j^*}(x,\tau^*) =0$ and $\omega_{i,j^*}(x,\tau^*) \geq0$. In this case, it follows from \eqref{lamda ij} that $\lambda^+(x) = \lambda_{i,j^*}(x,\tau^*) = 1$, which contradicts \eqref{relation2}. Therefore, this case cannot occur. 

Case (3): $\eta_{i,j^*}(x,\tau^*) =0$ and $\omega_{i,j^*}(x,\tau^*) <0$. It is straightforward to get that $\eta_{i,j^*}(x,\tau^*) (\lambda^+(x) + \Delta_\lambda) =0 > \omega_{i,j^*}(x,\tau^*)$. Similar to case (1), it can be proved that $\lambda^+(x) +\Delta_\lambda$ is infeasible. 

Case (4): $\eta_{i,j^*}(x,\tau^*) <0$. Since $\lambda_{i,j^*}(x,\tau^*)$ can take either 0 or 1 in this case and recalling that $\lambda_{i,j^*}(x,\tau^*) <1$, it must hold that $\lambda_{i,j^*}(x,\tau^*) =0$ and that $\omega_{i,j^*}(x,\tau^*) - \eta_{i,j^*}(x,\tau^*) \lambda_{i,j''}(x,\tau'') <0$, where $\lambda_{i,j''}(x,\tau'') = \min_{\tau' \in[0,\; T], j' \in \mathbb{N}^+_{\mathrm{dim}(\eta_{i})}, \lambda_{i,j'}(x,\tau') \in (0,1)}  \lambda_{i,j'} (x,\tau') $. As a consequence, if $\lambda^+(x) + \Delta_\lambda \leq \lambda_{i,j''}(x,\tau'')$, we have 
\begin{align}\label{42}
	\eta_{i,j^*}(x,\tau^*) (\lambda^+(x) + \Delta_\lambda) &\geq \eta_{i,j^*}(x,\tau^*) \lambda_{i,j''}(x,\tau'')\nonumber\\
    &>\omega_{i,j^*}(x,\tau^*).
\end{align}
If $\lambda^+(x) + \Delta_\lambda > \lambda_{i,j''}(x,\tau'')$, consider the $j''$-th component of the constraint of \eqref{lamda compact}. In particular, since $\lambda_{i,j''}(x,\tau'') \in (0,1)$, by \eqref{lamda ij} we know that $\omega_{i,j''}(x,\tau'') >0$, $\eta_{i,j''}(x,\tau'') >0$, and $\lambda_{i,j''}(x,\tau'') =\frac{\omega_{i,j''}(x,\tau'')}{\eta_{i,j''}(x,\tau'')}$. Then we have 
\begin{align}\label{43}
	\eta_{i,j''}(x,\tau'') (\lambda^+(x) + \Delta_\lambda) &> \eta_{i,j''}(x,\tau'') \lambda_{i,j''}(x,\tau'')\nonumber\\
    &>\omega_{i,j''}(x,\tau'').
\end{align}
Combining \eqref{42} and \eqref{43}, we observe that either the $j^*$-th row or the $j''$-th row of the constraint of \eqref{lamda compact} is violated for $\lambda^+(x)+\Delta_\lambda$ in Case (4).

Integrating the discussions in the four cases above and noticing the arbitrariness of $i$, we can conclude that $\lambda^+(x)+\Delta_\lambda$ is infeasible for problem \eqref{lamda compact}. 

\subsection{Proof of Lemma \ref{lem:verifying corollary}}\label{appendix:proof of verifying corollary}

{\color{blue}
Let $\rho_1=K/K_1$, $\rho_4=K/K_3$, $c_1=K+K_0+K_1$, $c_1'=K+K_0+K_2$,
$c_4=K+K_0+K_3$, $c_4'=K+K_0+K_4$, $a=2K+K_0$, $\beta_1=K(K_2/K_1-1)$,
$\beta_4=K(K_4/K_3-1)$, and $\bar K=5\times10^{-3}\,\mathrm{s}^{-1}$. With the data $K=0.0035$, $K_0=0.001$, $K_1=0.01$, $K_2=0.02$, $K_3=0.008$,
$K_4=0.016\,\mathrm{s}^{-1}$, one has $\rho_1=0.35$, $\rho_4=0.4375$,
$\beta_1=\beta_4=K$, and $K+K_0=0.0045$.

1) As $\mathcal{O}$ is a box and $f_\mathrm{PWA,b}$ is Lipschitz, by
Nagumo's theorem \cite{Nagumo1942} it suffices to show the vector field is sub-tangential on each face. 
Substituting $\pi_\mathrm{b}$ into the dynamics and using $T_0=0$, the backup closed loop
$\dot x=f_\mathrm{PWA,b}(x)$ reads
\begin{align*}
\dot T_1 &=
\begin{cases}
-(K+K_0+K_1)\,T_1, \text{ if } |KT_2/K_1+T_1|\le 5,\\[2pt]
K(T_2-T_1)-K_2\!\left(\tfrac{K}{K_1}T_2+T_1\right)-K_0T_1, \; \text{otherwise},
\end{cases}\nonumber\\
\dot T_2 &= K(T_1-T_2)+K(T_3-T_2)-K_0T_2,\nonumber\\
\dot T_3 &= K(T_2-T_3)+K(T_4-T_3)-K_0T_3,\nonumber\\
\dot T_4 &=
\begin{cases}
-(K+K_0+K_3)\,T_4, \text{ if } |KT_3/K_3+T_4|\le 5,\\[2pt]
K(T_3-T_4)-K_4\!\left(\tfrac{K}{K_3}T_3+T_4\right)-K_0T_4, \; \text{otherwise}.
\end{cases}\nonumber
\end{align*}

\emph{Faces $T_2=21$, $T_3=21$.} On $T_2=21$, using $T_1\le26$,
$T_3\le21$, we have $
\dot T_2=K(T_1+T_3)-aT_2\le 26K-(K+K_0)\,21=0.091-0.0945<0$. The face $T_3=21$ is symmetric.

\emph{Faces $T_1=26$, $T_4=26$.} At $T_1=26$, room~$1$ is in
mode~B (since $\rho_1T_2+T_1\ge26-\rho_1\cdot4.5>5$), so with $T_2\ge-4.5$,
$\dot T_1=-c_1'\cdot26-\beta_1T_2\le-26c_1'+4.5\,K=-0.117+0.016<0$;
symmetrically at $T_4=26$.

\emph{Faces $T_2=-4.5$, $T_3=-4.5$.} On $T_2=-4.5$, using
$T_1,T_3\ge-4.5$, we have
$\dot T_2=K(T_1+T_3)-aT_2\ge-2K(4.5)+a(4.5)=K_0(4.5)=4.5\times10^{-3}>0$
The face $T_3=-4.5$ is symmetric.

\emph{Faces $T_1=-4.5$, $T_4=-4.5$.} Since
$\rho_4\cdot21-5=4.19\le4.5$, for all admissible $T_3\le21$ we have
$\rho_4T_3-4.5\le\rho_4\cdot21-4.5\le5$, so the floor $T_4=-4.5$ never lies in the upper mode~B branch $\rho_4T_3+T_4>5$; the same holds for room~$1$ as
$\rho_1<\rho_4$. Two cases remain on $T_1=-4.5$. In mode~A,
$\dot T_1=-c_1T_1=c_1(4.5)>0$. In the lower mode~B branch
($\rho_1T_2+T_1<-5$, hence $T_2<0$),
$\dot T_1=-c_1'T_1-\beta_1T_2=c_1'(4.5)-\beta_1T_2>0$. Thus $\dot T_1\ge0$, and $\dot T_4\ge0$ follows identically.

Sub-tangentiality on every face gives forward invariance of $\mathcal{O}$,
proving (1).

2) Since $h_\mathrm{b}$ is a min of two
affine functions, $\partial_\mathrm{C}h_\mathrm{b}(x)=\{[-1,0,0,0]\}$ if $T_1>T_4$,
$\{[0,0,0,-1]\}$ if $T_1<T_4$, and $\partial_\mathrm{C}h_\mathrm{b}(x)$ equals their convex hull if $T_1=T_4$. Let $\alpha_b(a) = \bar K a$. As the inequality
is affine in $\partial h_\mathrm{b}$ and $h_\mathrm{b}=24-T_1=24-T_4$ on
$\{T_1=T_4\}$, it holds for every element of $\partial_\mathrm{C}h_\mathrm{b}$ iff it
holds at the two extreme gradients, i.e. iff
\begin{equation}\label{eq:cbf-two}
-\dot T_1\ge-\bar K(24-T_1)\quad\text{and}\quad -\dot T_4\ge-\bar K(24-T_4)
\end{equation}
on the respective sets $\{T_1\ge T_4\}\cap S_\mathrm{b}\cap \mathcal{O}$ and
$\{T_4\ge T_1\}\cap S_\mathrm{b}\cap \mathcal{O}$. We prove the first, and the second is symmetric.
On the first set $T_1\in[-4.5,24]$ and $T_2\ge-4.5$.

In mode~A of room 1, $-\dot T_1=c_1T_1$, so \eqref{eq:cbf-two} becomes
$24\bar K+(c_1-\bar K)T_1\ge0$; since $\bar K\le c_1$, the left side is minimized at
$T_1=-4.5$, requiring $24\bar K\ge(c_1-\bar K)4.5$. In mode~B,
$-\dot T_1=c_1'T_1+\beta_1T_2$, so \eqref{eq:cbf-two} becomes
$24\bar K+(c_1'-\bar K)T_1+\beta_1T_2\ge0$; as $c_1'\ge c_1\ge\bar K$ and
$\beta_1\ge0$, the minimum over $T_1,T_2\ge-4.5$ is at $T_1=T_2=-4.5$, requiring
$24\bar K\ge(c_1'-\bar K+\beta_1)4.5$. The mode~B condition dominates, and with the
data
\[
\frac{24\bar K}{c_1'+\beta_1-\bar K}=\frac{24(0.005)}{0.028-0.005}=5.22\ \ge\ 4.5,
\]
so it holds. For the room~$4$, the inequality in \eqref{eq:cbf-two} follows identically since
$c_4'+\beta_4=c_1'+\beta_1$ and $\bar K\le c_4$. This establishes
\eqref{eq:cbf-two} on $S_\mathrm{b}\cap \mathcal{O}$, proving (2). \qed}

\bibliographystyle{IEEEtran}
\bibliography{references} 

\end{document}